\documentclass[preprint,3p,times]{elsarticle}

\usepackage{amssymb}
\usepackage{amsmath}
\usepackage[all,cmtip]{xy}
\usepackage{latexsym}
\usepackage{amsthm}
\usepackage{color}
\usepackage[colorlinks=true]{hyperref}
\usepackage{stmaryrd}
\usepackage{yhmath}
\usepackage{tikz}
\usepackage[hyphenbreaks]{breakurl}
\input diagxy

\journal{}

\theoremstyle{plain}
  \newtheorem{thm}{Theorem}[section]
  \newtheorem{lem}[thm]{Lemma}
  \newtheorem{prop}[thm]{Proposition}
  \newtheorem{cor}[thm]{Corollary}
\theoremstyle{definition}

  \newtheorem{exmp}[thm]{Example}
  \newtheorem{rem}[thm]{Remark}

\DeclareMathAlphabet{\mathcal}{OMS}{cmsy}{m}{n}
\DeclareMathOperator*{\colim}{colim}

\DeclareMathOperator{\dom}{dom}
\DeclareMathOperator{\cod}{cod}
\DeclareMathOperator{\Lan}{Lan}
\DeclareMathOperator{\Ran}{Ran}

\DeclareMathOperator{\ob}{ob}

\makeatletter
\def\ps@pprintTitle{%
 \let\@oddhead\@empty
  \let\@evenhead\@empty
  \def\@oddfoot{\vbox{\hsize=\textwidth\footnotesize
  \vskip 8pt
  \copyright 2016. This manuscript version is made available under the CC-BY-NC-ND 4.0 license \url{http://creativecommons.org/licenses/by-nc-nd/4.0/}. The published version is available at \url{http://dx.doi.org/10.1016/j.fss.2016.12.001}.\\
  }}%
  \let\@evenfoot\@oddfoot}
\makeatother

\def\oto{{\bfig\morphism<180,0>[\mkern-4mu`\mkern-4mu;]\place(86,0)[\circ]\efig}}

\renewcommand{\phi}{\varphi}
\newcommand{\da}{\downarrow}

\newcommand{\ua}{\uparrow}
\newcommand{\ra}{\rightarrow}

\newcommand{\nra}{{\rightarrow\hspace*{-1ex}{\mapstochar}\hspace*{1.1ex}}}

\newcommand{\lra}{\longrightarrow}
\newcommand{\lda}{\swarrow}
\newcommand{\rda}{\searrow}

\newcommand{\Lra}{\Longrightarrow}

\newcommand{\bv}{\bigvee}
\newcommand{\bw}{\bigwedge}

\newcommand{\tr}{\triangleright}

\newcommand{\dv}{\dashv}

\newcommand{\nat}{\natural}
\newcommand{\ola}{\overleftarrow}
\newcommand{\ora}{\overrightarrow}

\newcommand{\lam}{\lambda}

\newcommand{\CD}{\mathcal{D}}

\newcommand{\CQ}{\mathcal{Q}}

\newcommand{\sK}{{\sf K}}
\newcommand{\sM}{{\sf M}}

\newcommand{\sP}{{\sf P}}

\newcommand{\sY}{{\sf Y}}

\newcommand{\sPd}{{\sP^{\dag}}}

\newcommand{\bbA}{\mathbb{A}}
\newcommand{\bbB}{\mathbb{B}}
\newcommand{\bbC}{\mathbb{C}}
\newcommand{\bbD}{\mathbb{D}}

\newcommand{\bbX}{\mathbb{X}}
\newcommand{\bbY}{\mathbb{Y}}

\newcommand{\obbA}{\overline{\bbA}}
\newcommand{\obbB}{\overline{\bbB}}
\newcommand{\obbAp}{\overline{\bbA'}}

\newcommand{\Cat}{{\bf Cat}}

\newcommand{\Chu}{{\bf Chu}}

\newcommand{\Dist}{{\bf Dist}}

\newcommand{\Set}{{\bf Set}}
\newcommand{\Sup}{{\bf Sup}}
\newcommand{\QCat}{\CQ\text{-}\Cat}

\newcommand{\QDist}{\CQ\text{-}\Dist}
\newcommand{\QChu}{\CQ\text{-}\Chu}

\newcommand{\QSup}{\CQ\text{-}\Sup}

\newcommand{\co}{{\rm co}}

\newcommand{\op}{{\rm op}}

\newcommand{\dphi}{\phi^{\da}}
\newcommand{\uphi}{\phi_{\ua}}

\newcommand{\dpsi}{\psi^{\da}}
\newcommand{\upsi}{\psi_{\ua}}
\newcommand{\olphi}{\ola{\phi}}

\newcommand{\ophi}{\ora{\phi}}

\newcommand{\PA}{\sP\bbA}
\newcommand{\PB}{\sP\bbB}

\newcommand{\PX}{\sP\bbX}

\newcommand{\PdX}{\sPd\bbX}

\newcommand{\PdA}{\sPd\bbA}
\newcommand{\PdB}{\sPd\bbB}

\newcommand{\sYd}{\sY^{\dag}}

\newcommand{\DL}{\CD L}

\newcommand{\Fix}{{\sf Fix}}
\newcommand{\sIm}{{\sf Im}}
\newcommand{\Mphi}{\sM\phi}
\newcommand{\Kphi}{\sK\phi}
\newcommand{\Nd}{N^{\dag}}
\newcommand{\Ud}{U^{\dag}}

\numberwithin{equation}{section}

\begin{document}

\begin{frontmatter}



\title{Fixed points of adjoint functors enriched in a quantaloid}


\author[S]{Hongliang Lai}
\ead{hllai@scu.edu.cn}

\author[Y]{Lili Shen\corref{cor}\fnref{C}}
\ead{shenlili@scu.edu.cn}

\cortext[cor]{Corresponding author.}
\address[S]{School of Mathematics, Sichuan University, Chengdu 610064, China}
\address[Y]{Department of Mathematics and Statistics, York University, Toronto, Ontario M3J 1P3, Canada}
\fntext[C]{Present address: School of Mathematics, Sichuan University, Chengdu 610064, China.}

\begin{abstract}
Representation theorems are established for fixed points of adjoint functors between categories enriched in a small quantaloid. In a very general setting these results set up a common framework for representation theorems of concept lattices in formal concept analysis (FCA) and rough set theory (RST), which not only extend the realm of formal contexts to multi-typed and multi-valued ones, but also provide a general approach to construct various kinds of representation theorems. Besides incorporating several well-known representation theorems in FCA and RST as well as formulating new ones, it is shown that concept lattices in RST can always be represented as those in FCA through relative pseudo-complements of the given contexts, especially if the contexts are valued in a non-Girard quantaloid.
\end{abstract}


\end{frontmatter}




\section{Introduction}

This paper aims to establish general representation theorems for fixed points of adjoint functors between categories enriched in a small quantaloid $\CQ$, which set up a common framework for representation theorems of concept lattices in formal concept analysis (FCA) \cite{Davey2002,Ganter1999} and rough set theory (RST) \cite{Pawlak1982,Polkowski2002} in the generality of their $\CQ$-version. As Galois connections between posets are precisely adjoint functors between categories enriched in the two-element Boolean algebra {\bf 2}, we start the introduction from this classical case.

A Galois connection \cite{Davey2002} $s\dv t$ between posets $C$, $D$ consists of monotone maps $s:C\lra D$, $t:D\lra C$ such that $s(x)\leq y\iff x\leq t(y)$ for all $x\in C$, $y\in D$. By a fixed point of $s\dv t$ is meant an element $x\in C$ with $x=ts(x)$ or, equivalently, an element $y\in D$ with $y=st(y)$, since
$$\Fix(ts):=\{x\in C\mid x=ts(x)\}\quad\text{and}\quad\Fix(st):=\{y\in D\mid y=st(y)\}$$
are isomorphic posets with the inherited order from $C$ and $D$, respectively. As the first main result of this paper, the following theorem characterizes those posets representing $\Fix(ts)\cong\Fix(st)$:

\begin{thm} \label{general_representation_poset}
Let $s\dv t:D\lra C$ be a Galois connection between posets. A poset $X$ is isomorphic to $\Fix(ts)$ if, and only if, there exist surjective maps $l:C\lra X$ and $r:D\lra X$ such that
$$\forall c\in C,\forall d\in D:\ s(c)\leq d\ \text{in}\ D\iff l(c)\leq r(d)\ \text{in}\ X.$$
\end{thm}

It is well known that if $C$, $D$ are complete lattices, then so is $\Fix(ts)\cong\Fix(st)$. In this case, the above representation theorem can be strengthened to the following one, which is our second main result, in terms of $\bv$-dense and $\bw$-dense maps:

\begin{thm} \label{general_representation_complete_lattice}
Let $s\dv t:D\lra C$ be a Galois connection between complete lattices. A complete lattice $X$ is isomorphic to $\Fix(ts)$ if, and only if, there exist $\bv$-dense maps $f:A\lra X$, $k:A\lra C$ and $\bw$-dense maps $g:B\lra X$, $h:B\lra D$ such that
$$\forall a\in A,\forall b\in B:\ sk(a)\leq h(b)\ \text{in}\ D\iff f(a)\leq g(b)\ \text{in}\ X.$$
\end{thm}
\begin{center}
$\bfig
\Vtriangle(-2500,-300)/@{>}@<4pt>`->`->/<500,600>[C`D`X;s`l`r]
\morphism(-1500,300)|b|/@{>}@<4pt>/<-1000,0>[D`C;t]
\place(-2000,305)[\mbox{\scriptsize$\bot$}]
\place(-2000,-450)[\text{Theorem \ref{general_representation_poset}}]
\morphism(-700,0)<400,300>[A`C;k]
\morphism(700,0)<-400,300>[B`D;h]
\morphism(-700,0)|b|<700,-300>[A`X;f]
\morphism(700,0)|b|<-700,-300>[B`X;g]
\morphism(-300,300)|a|/@{>}@<4pt>/<600,0>[C`D;s]
\morphism(300,300)|b|/@{>}@<4pt>/<-600,0>[D`C;t]
\place(0,305)[\mbox{\scriptsize$\bot$}]
\place(0,-450)[\text{Theorem \ref{general_representation_complete_lattice}}]
\efig$
\end{center}

These two theorems play the role of general representation theorems and their power will be revealed when being applied to concept lattices. To see this, recall that given a relation $\phi:A\oto B$ between sets (usually called a \emph{formal context}, or \emph{context} for short, and written as $(A,B,\phi)$ in FCA and RST), there are two Galois connections
\begin{equation} \label{uR_dv_dR}
\uphi\dv\dphi:({\bf 2}^B)^{\op}\lra{\bf 2}^A\quad\text{and}\quad\phi^*\dv\phi_*:{\bf 2}^A\lra{\bf 2}^B
\end{equation}
given by
$$\begin{tabular}{ll}
$\uphi(U)=\{y\in B\mid\forall x\in U:\ x\phi y\}$,&$\dphi(V)=\{x\in A\mid\forall y\in V:\ x\phi y\}$,\\
$\phi^*(V)=\{x\in A\mid\exists y\in V:\ x\phi y\}$,&$\phi_*(U)=\{y\in B\mid\forall x\in A:\ x\phi y{}\Lra{}x\in U\}$
\end{tabular}$$
for all $U\subseteq A$, $V\subseteq B$; the complete lattices consisting of their fixed points,
$$\Mphi:=\Fix(\dphi\uphi)\quad\text{and}\quad\Kphi:=\Fix(\phi_*\phi^*),$$
are respectively (up to isomorphism) the \emph{concept lattices} of the context $(A,B,\phi)$ in FCA and RST\footnote{$\Mphi$ and $\Kphi$ are also called the \emph{formal concept lattice} and the \emph{object-oriented concept lattice} of the context $(A,B,\phi)$, respectively.}. The \emph{fundamental theorem} of FCA characterizes those complete lattices representing $\Mphi$:

\begin{thm} \label{FCA_fundamental} (See \cite{Davey2002,Ganter1999}.)
A complete lattice $X$ is isomorphic to $\Mphi$ if, and only if, there exist a $\bv$-dense map $f:A\lra X$ and a $\bw$-dense map $g:B\lra X$ such that
$$\forall a\in A,\ \forall b\in B:\ a\phi b\iff f(a)\leq g(b)\ \text{in}\ X.$$
\end{thm}

The following diagrams explain how one derives the above theorem from \ref{general_representation_poset} and \ref{general_representation_complete_lattice}:
$$\bfig
\morphism(-800,0)<500,300>[A`{\bf 2}^A;k=\{\text{-}\}_A]
\morphism(800,0)<-500,300>[B`({\bf 2}^B)^{\op};h=\{\text{-}\}_B]
\morphism(-800,0)|b|<800,-300>[A`\Mphi;f=lk]
\morphism(800,0)|b|<-800,-300>[B`\Mphi;g=rh]
\morphism(-300,300)|a|/@{>}@<4pt>/<600,0>[{\bf 2}^A`({\bf 2}^B)^{\op};\uphi]
\morphism(300,300)|b|/@{>}@<4pt>/<-600,0>[({\bf 2}^B)^{\op}`{\bf 2}^A;\dphi]
\morphism(-300,300)|l|<300,-600>[{\bf 2}^A`\Mphi;l]
\morphism(300,300)|r|<-300,-600>[({\bf 2}^B)^{\op}`\Mphi;r]
\place(-40,305)[\mbox{\scriptsize$\bot$}]
\morphism(1100,0)<600,300>[A`{\bf 2}^A;k=\{\text{-}\}_A]
\morphism(2900,0)<-600,300>[B`({\bf 2}^B)^{\op};h=\{\text{-}\}_B]
\morphism(1100,0)|b|<900,-300>[A`X;f]
\morphism(2900,0)|b|<-900,-300>[B`X;g]
\morphism(1700,300)|a|/@{>}@<4pt>/<600,0>[{\bf 2}^A`({\bf 2}^B)^{\op};\uphi]
\morphism(2300,300)|b|/@{>}@<4pt>/<-600,0>[({\bf 2}^B)^{\op}`{\bf 2}^A;\dphi]
\morphism(1700,300)|l|<300,-600>[{\bf 2}^A`X;l=\Lan_k f]
\morphism(2300,300)|r|<-300,-600>[({\bf 2}^B)^{\op}`X;r=\Ran_h g]
\place(1960,305)[\mbox{\scriptsize$\bot$}]
\efig$$

Explicitly, for the ``only if'' part of Theorem \ref{FCA_fundamental} it suffices to consider $X=\Mphi$; one has a $\bv$-dense map $\{\text{-}\}_A:A\lra{\bf 2}^A$ sending each $x\in A$ to the singleton set $\{x\}\subseteq A$ whose composition with $l:{\bf 2}^A\lra\Mphi$ from Theorem \ref{general_representation_poset} gives the required $\bv$-dense map $f:A\lra\Mphi$, and the $\bw$-dense map $g:B\lra\Mphi$ is constructed dually. Conversely, for the ``if'' part the $\bv$-dense maps $f:A\lra X$, $\{\text{-}\}_A:A\lra{\bf 2}^A$ and $\bw$-dense maps $g:B\lra X$, $\{\text{-}\}_B:B\lra({\bf 2}^B)^{\op}$ fulfill the requirements of Theorem \ref{general_representation_complete_lattice}; indeed, one may further show that the \emph{left Kan extension} $l:{\bf 2}^A\lra X$ of $f$ along $k$ \cite{MacLane1998} and the \emph{right Kan extension} $r:({\bf 2}^B)^{\op}\lra X$ of $g$ along $h$ (see Subsection \ref{Kan_extensions}) satisfy the requirements of Theorem \ref{general_representation_poset}.

As for the representation of $\Kphi$, since it is well known that $\Kphi=\sM(\neg\phi)$ \cite{Duntsch2002,Yao2004}, where $\neg\phi:B\oto A$ is the \emph{complement} of the relation $\phi:A\oto B$ given by
$$\forall b\in B,\forall a\in A:\ b(\neg\phi)a\iff \neg(a\phi b),$$
the following theorem easily follows from \ref{FCA_fundamental}:

\begin{thm} \label{RST_fundamental}
A complete lattice $X$ is isomorphic to $\Kphi$ if, and only if, there exist a $\bv$-dense map $f:B\lra X$ and a $\bw$-dense map $g:A\lra X$ such that
$$\forall b\in B,\ \forall a\in A:\ b(\neg\phi) a\iff f(b)\leq g(a)\ \text{in}\ X.$$
\end{thm}

As (unital) \emph{quantales} are usually chosen as truth tables in fuzzy set theory, Galois connections have been extended to the quantale-valued setting \cite{Belohlavek1999,Garcia2010} as well as the theories of FCA and RST \cite{Bvelohlavek2001,Bvelohlavek2004,Georgescu2004,Lai2009,Popescu2004}, and all the representation theorems stated above can be established in this general setting. Since Galois connections between quantale-valued ordered sets are precisely adjoint functors between quantale-enriched categories, in fact, we will present these theorems in an even more general framework of adjoint functors between categories enriched in a small \emph{quantaloid} $\CQ$.

Quantaloids \cite{Rosenthal1996} may be thought of as quantales with many objects; indeed, let $\Sup$ denote the symmetric monoidal closed category of complete lattices and $\sup$-preserving maps, then a quantale is a \emph{monoid in $\Sup$} while a quantaloid is a \emph{category enriched in $\Sup$}. The theory of quantaloid-enriched categories (or \emph{$\CQ$-categories} for short), as an extension of quantale-enriched categories \cite{Kelly1982,Lai2007,Lawvere1973}, has been developed in \cite{Heymans2010,Rosenthal1996,Stubbe2005,Stubbe2006}; the survey paper \cite{Stubbe2014} is particularly recommended as an overview of this theory for the readership of fuzzy logicians and fuzzy set theorists.

We recall the basics of quantaloid-enriched categories in Section \ref{Quantaloid-enriched categories} and present Kan extensions and (co)dense $\CQ$-functors as our key tools in Section \ref{Colimits_Kan_codense}. Under this general framework, Sections \ref{General_Representation} and \ref{General_Representation_Dense} are respectively devoted to establishing our main results, Theorems \ref{general_representation_poset} and \ref{general_representation_complete_lattice}, in the generality of their $\CQ$-version:
\begin{itemize}
\item {\bf Theorem \ref{general_representation}}: Let $S\dv T:\bbD\lra\bbC$ be an adjunction in $\QCat$. A $\CQ$-category $\bbX$ is equivalent to $\Fix(TS)$ if, and only if, there exist essentially surjective $\CQ$-functors $L:\bbC\lra\bbX$ and $R:\bbD\lra\bbX$ with $\bbD(S-,-)=\bbX(L-,R-)$.
\item {\bf Theorem \ref{general_representation_dense_complete}}: Let $S\dv T:\bbD\lra\bbC$ be an adjunction between complete $\CQ$-categories. Then a complete $\CQ$-category $\bbX$ is equivalent to $\Fix(TS)$ if, and only if, there exist dense $\CQ$-functors $F:\bbA\lra\bbX$, $K:\bbA\lra\bbC$ and codense $\CQ$-functors $G:\bbB\lra\bbX$, $H:\bbB\lra\bbD$ with $\bbD(SK-,H-)=\bbX(F-,G-)$.
    \begin{center}
    $\bfig
    \Vtriangle(-2500,-300)/@{>}@<4pt>`->`->/<500,600>[\bbC`\bbD`\bbX;S`L`R]
    \morphism(-1500,300)|b|/@{>}@<4pt>/<-1000,0>[\bbD`\bbC;T]
    \place(-2000,305)[\mbox{\scriptsize$\bot$}]
    \place(-2000,-450)[\text{Theorem \ref{general_representation}}]
    \morphism(-700,0)<400,300>[\bbA`\bbC;K]
    \morphism(700,0)<-400,300>[\bbB`\bbD;H]
    \morphism(-700,0)|b|<700,-300>[\bbA`\bbX;F]
    \morphism(700,0)|b|<-700,-300>[\bbB`\bbX;G]
    \morphism(-300,300)|a|/@{>}@<4pt>/<600,0>[\bbC`\bbD;S]
    \morphism(300,300)|b|/@{>}@<4pt>/<-600,0>[\bbD`\bbC;T]
    \place(0,305)[\mbox{\scriptsize$\bot$}]
    \place(0,-450)[\text{Theorem \ref{general_representation_dense_complete}}]
    \efig$
    \end{center}
\end{itemize}
In fact, Theorems \ref{general_representation_poset} and \ref{general_representation_complete_lattice} are respectively special cases of Corollary \ref{general_representation_type_preserving} and Theorem \ref{general_representation_elementary} when $\CQ={\bf 2}$, which will be shown later as immediate consequences of Theorems \ref{general_representation} and \ref{general_representation_dense_complete}.

The applications of the general representation theorems in FCA and RST are discussed in Section \ref{Isbell_Kan_fixed_points}. Note that distributors between $\CQ$-categories (or \emph{$\CQ$-distributors} for short) generalize relations between sets in the sense that a $\CQ$-distributor may be thought of as a multi-typed and multi-valued relation which respects the $\CQ$-categorical structures in its domain and codomain. Thus, $\CQ$-distributors may be considered as \emph{multi-typed and multi-valued contexts} upon which a general theory of FCA and RST can be established (see \cite[Section 4]{Shen2016a} for instance).

Explicitly, each $\CQ$-distributor $\phi:\bbA\oto\bbB$ induces two pairs of adjoint $\CQ$-functors between the (co)presheaf $\CQ$-categories of $\bbA$ and $\bbB$, i.e.,
$$\uphi\dv\dphi:\PdB\lra\PA\quad\text{and}\quad\phi^*\dv\phi_*:\PA\lra\PB,$$
called respectively the \emph{Isbell adjunction} and \emph{Kan adjunction} \cite{Shen2013a} induced by $\phi$, whose fixed points constitute complete $\CQ$-categories $\Mphi$ and $\Kphi$, respectively. As our notations already suggest, Isbell adjunctions and Kan adjunctions induced by $\CQ$-distributors present the $\CQ$-categorical version of the Galois connections \eqref{uR_dv_dR} in FCA and RST. Therefore, for a $\CQ$-distributor $\phi:\bbA\oto\bbB$, $\Mphi$ and $\Kphi$ may be respectively viewed as ``concept lattices'' of the multi-typed and multi-valued context $(\bbA,\bbB,\phi)$ in FCA and RST.

Although it is straightforward to extend Theorem \ref{FCA_fundamental} to the $\CQ$-version (see Theorem \ref{Mphi_representation}), the validity of Theorem \ref{RST_fundamental} relies heavily on the fact that {\bf 2}, as a Boolean algebra, satisfies the law of double negation, which guarantees the existence of the complement $\neg\phi$. For a quantaloid $\CQ$, the existence of $\neg\phi$ requires $\CQ$ to be a \emph{Girard quantaloid} \cite{Rosenthal1992} (an extension of \emph{Girard quantales} \cite{Rosenthal1990,Yetter1990}). In fact, it is impossible to extend Theorem \ref{RST_fundamental} directly to the $\CQ$-version without assuming $\CQ$ being Girard: as Lai-Zhang revealed in the case that $\CQ$ is a commutative integral quantale\footnote{An \emph{integral} quantale is a unital quantale in which the unit is the top element of the quantale.} (see \cite[Proposition 5.5]{Lai2009}), in general a codense $\CQ$-functor $\bbA\lra\Kphi$ may not even exist! This observation can be extended to a quantaloid $\CQ$ with some mild assumptions (Proposition \ref{codense_Girard}), which reveals that even the existence of codense $\CQ$-functors $\bbA\lra\PA$ would require $\CQ$ to be Girard.

Hence, for a general quantaloid $\CQ$, in order to apply Theorem \ref{general_representation_dense_complete} to $\Kphi$, one needs to find a non-trivial codense $\CQ$-subcategory of $\PA$ which would unavoidably have a larger size than $\bbA$. To this end, we construct a $\CQ$-subcategory $\obbA$ of $\PA$ consisting of all the possible \emph{relative pseudo-complements of representable copresheaves} on $\bbA$. Then, by defining the $\CQ$-distributor $\bbA^{\tr}:\bbA\oto\obbA$ as the codomain restriction of the graph of the Yoneda embedding $(\sY_{\bbA})_{\nat}:\bbA\oto\PA$, one has the \emph{relative pseudo-complement}
$$\phi^{\tr}:=\bbA^{\tr}\lda\phi$$
of any $\CQ$-distributor $\phi:\bbA\oto\bbB$ with respect to $\bbA^{\tr}$, through which the precise condition of a complete $\CQ$-category representing $\Kphi$ is obtained (Theorem \ref{Kphi_representation}). Indeed, we prove
\begin{equation} \label{Kphi=Mphitr}
\Kphi=\sM\phi^{\tr}
\end{equation}
in the proof of \ref{Kphi_representation}, which represents the ``concept lattice'' of any multi-typed and multi-valued context in RST as the ``concept lattice'' of the relative pseudo-complement of the given context in FCA. Furthermore, the identity \eqref{Kphi=Mphitr} can be established on the functorial level as Proposition \ref{Kphi=Mphitr_functor} reveals.

Finally, Theorem \ref{general_representation_elementary} is presented in Section \ref{Elementary} as an elementary representation theorem for fixed points of adjoint $\CQ$-functors in terms of order-theoretic notions, i.e., $\bv$-dense and $\bw$-dense maps. By the aid of this theorem one is able to incorporate B{\v e}lohl{\' a}vek's representation theorem for concept lattices of quantale-valued contexts in FCA \cite[Theorem 14(2)]{Bvelohlavek2004} and Popescu's representation theorem for those in RST \cite[Proposition 7.3]{Popescu2004} into our general framework (see Corollary \ref{Mphi_Kphi_representation_quantale} and Remark \ref{Belohlavek_Popescu}). In fact, their results are extended to the quantaloid-enriched version (Theorems \ref{Mphi_representation_elementary} and \ref{Kphi_representation_elementary}) which outline the difference between the representations of $\Mphi$ and $\Kphi$; as Corollary \ref{Mphi_Kphi_representation_quantale} shows, this subtle distinction could easily be ignored when $\CQ$ is a quantale.

\section{Quantaloid-enriched categories} \label{Quantaloid-enriched categories}

A \emph{quantaloid} \cite{Rosenthal1996} $\CQ$ is a locally ordered 2-category whose hom-sets are complete lattices such that the composition $\circ$ of arrows preserves joins in each variable. The corresponding adjoints induced by the compositions
\begin{align*}
-\circ u\dv -\lda u:&\ \CQ(p,r)\lra\CQ(q,r),\\
v\circ -\dv v\rda -:&\ \CQ(p,r)\lra\CQ(p,q)
\end{align*}
satisfy
$$v\circ u\leq w\iff v\leq w\lda u\iff u\leq v\rda w$$
for all $\CQ$-arrows $u:p\lra q$, $v:q\lra r$, $w:p\lra r$, where the operations $\lda$, $\rda$ are called \emph{left} and \emph{right implications} in $\CQ$, respectively.

Unless otherwise specified, throughout this paper $\CQ$ denotes a small quantaloid with a set $\CQ_0$ of objects and a set $\CQ_1$ of arrows. The identity $\CQ$-arrow on $q\in\CQ_0$ will be denoted by $1_q$.

Considering $\CQ_0$ as a ``base'' set, a \emph{$\CQ$-typed set} is a set $A$ equipped with a map $|\text{-}|:A\lra\CQ_0$ sending each $x\in A$ to its \emph{type} $|x|\in\CQ_0$. A map $F:A\lra B$ between $\CQ$-typed sets is \emph{type-preserving} if $|x|=|Fx|$ for all $x\in A$. $\CQ$-typed sets and type-preserving maps constitute the slice category $\Set\da\CQ_0$.

A $\CQ$-category $\bbA$ consists of a $\CQ$-typed set $\bbA_0$ and hom-arrows $\bbA(x,y)\in\CQ(|x|,|y|)$ for all $x,y\in\bbA_0$ such that $1_{|x|}\leq\bbA(x,x)$ and $\bbA(y,z)\circ\bbA(x,y)\leq\bbA(x,z)$ for all $x,y,z\in\bbA_0$. A $\CQ$-category $\bbB$ is a \emph{$\CQ$-subcategory} of $\bbA$ if $\bbB_0\subseteq\bbA_0$ and $\bbB(x,y)=\bbA(x,y)$ for all $x,y\in\bbB_0$.

Each $\CQ$-category $\bbA$ admits a natural underlying (pre)order on $\bbA_0$ given by $x\leq y$ if $|x|=|y|$ and $1_{|x|}\leq\bbA(x,y)$. A $\CQ$-category $\bbA$ is \emph{separated} (or \emph{skeletal}) if $x\cong y$ (i.e., $x\leq y$ and $y\leq x$) implies $x=y$ for all $x,y\in\bbA_0$.

A \emph{$\CQ$-distributor} $\phi:\bbA\oto\bbB$ between $\CQ$-categories is given by a family of $\CQ$-arrows $\{\phi(x,y):|x|\lra|y|\}_{x\in\bbA_0,y\in\bbB_0}$ such that $\bbB(y,y')\circ\phi(x,y)\circ\bbA(x',x)\leq\phi(x',y')$ for all $x,x'\in\bbA_0$, $y,y'\in\bbB_0$. With the pointwise local order inherited from $\CQ$, $\CQ$-categories and $\CQ$-distributors constitute a (large) quantaloid $\QDist$ in which
\begin{align*}
&\psi\circ\phi:\bbA\oto\bbC,\quad(\psi\circ\phi)(x,z)=\bv_{y\in\bbB_0}\psi(y,z)\circ\phi(x,y),\\
&\xi\lda\phi:\bbB\oto\bbC,\quad(\xi\lda\phi)(y,z)=\bw_{x\in\bbA_0}\xi(x,z)\lda\phi(x,y),\\
&\psi\rda\xi:\bbA\oto\bbB,\quad (\psi\rda\xi)(x,y)=\bw_{z\in\bbC_0}\psi(y,z)\rda\xi(x,z)
\end{align*}
for $\CQ$-distributors $\phi:\bbA\oto\bbB$, $\psi:\bbB\oto\bbC$, $\xi:\bbA\oto\bbC$; the identity $\CQ$-distributor on $\bbA$ is given by hom-arrows $\bbA:\bbA\oto\bbA$.

A \emph{$\CQ$-functor} $F:\bbA\lra\bbB$ between $\CQ$-categories is a type-preserving map $F:\bbA_0\lra\bbB_0$ with $\bbA(x,x')\leq\bbB(Fx,Fx')$ for all $x,x'\in\bbA_0$. With the pointwise (pre)order of $\CQ$-functors given by
$$F\leq G:\bbA\lra\bbB\iff\forall x\in\bbA_0:\ Fx\leq Gx\iff\forall x\in\bbA_0:\ 1_{|x|}\leq\bbB(Fx,Gx),$$
$\CQ$-categories and $\CQ$-functors constitute a 2-category $\QCat$.

\begin{rem}  \label{Qcat_dual}
The \emph{dual} of a $\CQ$-category $\bbA$ is a $\CQ^{\op}$-category\footnote{The terminologies adopted here are not exactly the same as in the references \cite{Stubbe2005,Stubbe2006,Stubbe2007,Stubbe2014}: Our $\CQ$-categories are exactly $\CQ^{\op}$-categories in the sense of Stubbe.}, given by $\bbA^{\op}_0=\bbA_0$ and $\bbA^{\op}(x,y)=\bbA(y,x)$ for all $x,y\in\bbA_0$. Each $\CQ$-functor $F:\bbA\lra\bbB$ becomes a $\CQ^{\op}$-functor $F^{\op}:\bbA^{\op}\lra\bbB^{\op}$ with the same mapping on objects but $(F')^{\op}\leq F^{\op}$ whenever $F\leq F':\bbA\lra\bbB$. Each $\CQ$-distributor $\phi:\bbA\oto\bbB$ corresponds bijectively to a $\CQ^{\op}$-distributor $\phi^{\op}:\bbB^{\op}\oto\bbA^{\op}$ with $\phi^{\op}(y,x)=\phi(x,y)$ for all $x\in\bbA_0$, $y\in\bbB_0$. Therefore, as already noted in \cite{Stubbe2005}, one has a 2-isomorphism
\begin{equation} \label{QopCat_iso}
(-)^{\op}:\QCat\cong(\CQ^{\op}\text{-}\Cat)^{\co}
\end{equation}
and an isomorphism of quantaloids
\begin{equation} \label{QopDist_iso}
(-)^{\op}:\QDist\cong(\CQ^{\op}\text{-}\Dist)^{\op},
\end{equation}
where ``$\co$'' refers to reversing order in hom-sets.
\end{rem}

Each $\CQ$-functor $F:\bbA\lra\bbB$ induces a pair of $\CQ$-distributors given by
$$F_{\nat}:\bbA\oto\bbB,\quad F_{\nat}(x,y)=\bbB(Fx,y)\quad\text{and}\quad F^{\nat}:\bbB\oto\bbA,\quad F^{\nat}(y,x)=\bbB(y,Fx),$$
called respectively the \emph{graph} and \emph{cograph} of $F$, which form an adjunction $F_{\nat}\dv F^{\nat}$ in the 2-category $\QDist$, i.e., $\bbA\leq F^{\nat}\circ F_{\nat}$ and $F_{\nat}\circ F^{\nat}\leq\bbB$. It is easy to see
\begin{equation} \label{F_leq_G_graph}
F\leq G:\bbA\lra\bbB\iff G_{\nat}\leq F_{\nat}:\bbA\oto\bbB\iff F^{\nat}\leq G^{\nat}:\bbB\oto\bbA,
\end{equation}
and therefore
\begin{equation} \label{graph_cograph_functorial}
(-)_{\nat}:\QCat\lra(\QDist)^{\co},\quad(-)^{\nat}:\QCat\lra(\QDist)^{\op}
\end{equation}
are both 2-functors.

It is straightforward to verify the following propositions:

\begin{prop} \label{distributor_graph}
For any $\CQ$-distributor $\phi:\bbA\oto\bbB$ and $\CQ$-functors $F:\bbX\lra\bbA$, $G:\bbY\lra\bbB$,
$$\phi(F-,G-)=G^{\nat}\circ\phi\circ F_{\nat}=G_{\nat}\rda(\phi\lda F^{\nat}).$$
\end{prop}

\begin{prop} (See \cite{Shen2013a}.) \label{functor_graph}
Let $F:\bbA\lra\bbB$ be a $\CQ$-functor.
\begin{enumerate}[\rm (1)]
\item $\bbA=F^{\nat}\circ F_{\nat}$ if, and only if, $F$ is \emph{fully faithful} in the sense that $\bbA(x,y)=\bbB(Fx,Fy)$ for all $x,y\in\bbA_0$.
\item If $F$ is \emph{essentially surjective} in the sense that there exists $x\in\bbA_0$ with $Fx\cong y$ for all $y\in\bbB_0$, then $F_{\nat}\circ F^{\nat}=\bbB$.
\end{enumerate}
\end{prop}

\begin{prop} (See \cite{Heymans2010}.) \label{adjoint_arrow_calculation}
The following identities hold for all $\CQ$-functors $F$ and $\CQ$-distributors $\phi$, $\psi$ whenever the operations make sense:
\begin{enumerate}[\rm (1)]
\item $\phi\circ F_{\nat}=\phi\lda F^{\nat}$,\quad $F^{\nat}\circ\phi=F_{\nat}\rda\phi$.
\item $(F_{\nat}\circ\phi)\rda\psi=\phi\rda(F^{\nat}\circ\psi)$,\quad $(\psi\circ F_{\nat})\lda\phi=\psi\lda(\phi\circ F^{\nat})$.
\item $(\phi\rda\psi)\circ F_{\nat}=\phi\rda(\psi\circ F_{\nat})$,\quad $F^{\nat}\circ(\psi\lda\phi)=(F^{\nat}\circ\psi)\lda\phi$.
\item $F^{\nat}\circ(\phi\rda\psi)=(\phi\circ F_{\nat})\rda\psi$,\quad $(\psi\lda\phi)\circ F_{\nat}=\psi\lda(F^{\nat}\circ\phi)$.
\end{enumerate}
\end{prop}

A $\CQ$-functor $F:\bbA\lra\bbB$ is an \emph{equivalence} (resp. \emph{isomorphism}) of $\CQ$-categories if there exists a $\CQ$-functor $G:\bbB\lra\bbA$ with $GF\cong 1_{\bbA}$ and $FG\cong 1_{\bbB}$ (resp. $GF=1_{\bbA}$ and $FG=1_{\bbB}$), where $1_{\bbA}$ and $1_{\bbB}$ respectively denote the identity $\CQ$-functors on $\bbA$ and $\bbB$. In this case, we write $\bbA\simeq\bbB$ (resp. $\bbA\cong\bbB$) to denote that $\bbA$ and $\bbB$ are equivalent (resp. isomorphic) $\CQ$-categories.

\begin{prop} (See \cite{Stubbe2005}.)
A $\CQ$-functor is an equivalence (resp. isomorphism) of $\CQ$-categories if, and only if, it is fully faithful and essentially surjective (resp. fully faithful and bijective).
\end{prop}

A pair of $\CQ$-functors $F:\bbA\lra\bbB$, $G:\bbB\lra\bbA$ forms an adjunction $F\dv G:\bbB\lra\bbA$ in $\QCat$ if $1_{\bbA}\leq GF$ and $FG\leq 1_{\bbB}$. It is easy to obtain the following equivalent characterizations of adjoint $\CQ$-functors:

\begin{prop} (See \cite{Stubbe2005}.) \label{adjoint_graph}
Let $F:\bbA\lra\bbB$, $G:\bbB\lra\bbA$ be a pair of $\CQ$-functors. Then
$$F\dv G\ \text{in}\ \QCat\iff F_{\nat}=G^{\nat}\iff G_{\nat}\dv F_{\nat}\ \text{in}\ \QDist\iff G^{\nat}\dv F^{\nat}\ \text{in}\ \QDist.$$
\end{prop}


\section{Fixed points of adjoint $\CQ$-functors and their representation} \label{General_Representation}

For a $\CQ$-functor $F:\bbA\lra\bbA$, an object $x\in\bbA_0$ is a \emph{fixed point} of $F$ if $Fx\cong x$, and we denote by $\Fix(F)$ the $\CQ$-subcategory of $\bbA$ consisting of fixed points of $F$.

A \emph{$\CQ$-closure operator} \cite{Shen2013a} on a $\CQ$-category $\bbA$ is a $\CQ$-functor $F:\bbA\lra\bbA$ with $1_{\bbA}\leq F$ and $FF\cong F$.

\begin{prop} \label{Q_closure_la} (See \cite{Shen2013a}.)
For each $\CQ$-closure operator $F:\bbA\lra\bbA$, the inclusion $\CQ$-functor $\Fix(F)\ \to/^(->/\bbA$ is right adjoint to the codomain restriction $F:\bbA\lra\Fix(F)$.
\end{prop}

\begin{rem}
In the language of category theory, a $\CQ$-closure operator $F:\bbA\lra\bbA$ is a \emph{$\CQ$-monad} on $\bbA$ (note that the ``$\CQ$-natural transformation'' between $\CQ$-functors is simply given by the local order in $\QCat$), and objects in $\Fix(F)$ are precisely Eilenberg-Moore algebras of this $\CQ$-monad.
\end{rem}

Dually, \emph{$\CQ$-interior operators} correspond bijectively to $\CQ^{\op}$-closure operators under the isomorphism \eqref{QopCat_iso} in Remark \ref{Qcat_dual}; that is, $\CQ$-functors $F:\bbA\lra\bbA$ with $F\leq 1_{\bbA}$ and $FF\cong F$. The dual of Proposition \ref{Q_closure_la} states precisely that for each $\CQ$-interior operator $F:\bbA\lra\bbA$, the inclusion $\CQ$-functor $\Fix(F)\ \to/^(->/\bbA$ is left adjoint to the codomain restriction $F:\bbA\lra\Fix(F)$.

Each adjunction $S\dv T:\bbD\lra\bbC$ in $\QCat$ gives rise to a $\CQ$-closure operator $TS:\bbC\lra\bbC$ and a $\CQ$-interior operator $ST:\bbD\lra\bbD$. It is easy to see that the restrictions of $S$ and $T$,
\begin{align*}
S:\Fix(TS)\lra\Fix(ST)\quad\text{and}\quad T:\Fix(ST)\lra\Fix(TS),
\end{align*}
establish an equivalence of $\CQ$-categories, thus objects in both $\Fix(TS)$ and $\Fix(ST)$ will be referred to as fixed points of the adjoint $\CQ$-functors $S\dv T$. The following theorem describes those $\CQ$-categories which represent the fixed points of $S\dv T$:

\begin{thm} \label{general_representation}
Let $S\dv T:\bbD\lra\bbC$ be an adjunction in $\QCat$. A $\CQ$-category $\bbX$ is equivalent to $\Fix(TS)$ if, and only if, there exist essentially surjective $\CQ$-functors $L:\bbC\lra\bbX$ and $R:\bbD\lra\bbX$ with $S_{\nat}=R^{\nat}\circ L_{\nat}$.
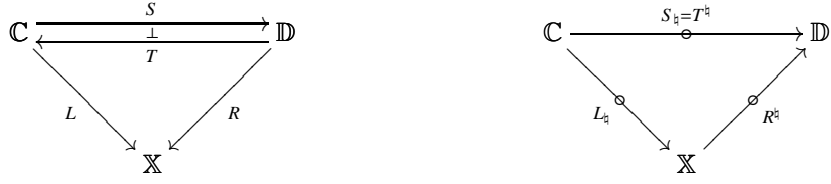
$$\bfig
\Vtriangle/@{>}@<4pt>`->`->/[\bbC`\bbD`\bbX;S`L`R]
\morphism(1000,500)|b|/@{>}@<3pt>/<-1000,0>[\bbD`\bbC;T]
\place(500,510)[\mbox{\scriptsize$\bot$}]
\Vtriangle(2000,0)/->`->`<-/[\bbC`\bbD`\bbX;S_{\nat}=T^{\nat}`L_{\nat}`R^{\nat}]
\place(2250,250)[\circ] \place(2750,250)[\circ] \place(2500,500)[\circ]
\efig$$
\end{thm}

\begin{proof}
{\bf Necessity.} It suffices to prove the case $\bbX=\Fix(TS)$. Let $L:\bbC\lra\Fix(TS)$ and $R:\bbD\lra\Fix(TS)$ be the codomain restriction of $TS:\bbC\lra\bbC$ and $T:\bbD\lra\bbC$, respectively, then $L$ and $R$ are clearly essentially surjective and satisfy
\begin{align*}
S_{\nat}&=(STS)_{\nat}&(S\dv T)\\
&=T^{\nat}\circ(TS)_{\nat}&(S\dv T\ \text{and Proposition \ref{adjoint_graph}})\\
&=\bbC(TS-,T-)&(\text{Proposition \ref{distributor_graph}})\\
&=\Fix(TS)(L-,R-)\\
&=R^{\nat}\circ L_{\nat}.&(\text{Proposition \ref{distributor_graph}})
\end{align*}

{\bf Sufficiency.} We show that the restriction $L':\Fix(TS)\lra\bbX$ of $L$ is an equivalence of $\CQ$-categories.

First, $LT\cong R$ and $RS\cong L$. Indeed, by Propositions \ref{functor_graph}(2) one has
$$R^{\nat}=R^{\nat}\circ L_{\nat}\circ L^{\nat}=S_{\nat}\circ L^{\nat}=T^{\nat}\circ L^{\nat}=(LT)^{\nat}\quad\text{and}\quad L_{\nat}=R_{\nat}\circ R^{\nat}\circ L_{\nat}=R_{\nat}\circ S_{\nat}=(RS)_{\nat}.$$
Thus the conclusion follows from \eqref{F_leq_G_graph}.

Second, $L'$ is fully faithful since for all $c,c'\in\Fix(TS)$,
\begin{align*}
\bbX(L'c,L'c')&=\bbX(Lc,LTSc')&(c'\cong TSc')\\
&=\bbX(Lc,RSc')&(LT\cong R)\\
&=S_{\nat}(c,Sc')&(S_{\nat}=R^{\nat}\circ L_{\nat}=\bbX(L-,R-))\\
&=T^{\nat}(c,Sc')&(T^{\nat}=S_{\nat})\\
&=\bbC(c,TSc')\\
&=\Fix(TS)(c,c').&(c'\cong TSc')
\end{align*}

Finally, $L'$ is essentially surjective since for any $c\in\bbC_0$, $RS\cong L$ and $S\dv T$ imply
\begin{equation} \label{L=LTS}
Lc\cong RSc\cong RSTSc\cong L(TSc)=L'(TSc).
\end{equation}
Hence the essential surjectivity of $L:\bbC\lra\bbX$ implies that of $L'$, completing the proof.
\end{proof}

From \eqref{L=LTS} in the above proof one sees that $L$, up to isomorphism, is the composition of an equivalence $L':\Fix(TS)\lra\bbX$ and a left adjoint $TS:\bbC\lra\Fix(TS)$ (see Proposition \ref{Q_closure_la}), thus $L$ itself must be a left adjoint in $\QCat$. Similarly one may deduce that $R$ is a right adjoint in $\QCat$:

\begin{cor} \label{L_la_R_ra}
The $\CQ$-functors $L$ and $R$ in Theorem \ref{general_representation} are respectively a left adjoint and a right adjoint in $\QCat$.
\end{cor}

The condition given in Theorem \ref{general_representation} can be weakened as in the following corollary since the $\CQ$-functoriality of $L$ and $R$ is self-contained:

\begin{cor} \label{general_representation_type_preserving}
Let $S\dv T:\bbD\lra\bbC$ be an adjunction in $\QCat$. A $\CQ$-category $\bbX$ is equivalent to $\Fix(TS)$ if, and only if, there exist essentially surjective type-preserving maps $L:\bbC_0\lra\bbX_0$ and $R:\bbD_0\lra\bbX_0$ with $S_{\nat}=\bbX(L-,R-)$.
\end{cor}

\begin{proof}
For all $c,c'\in\bbC_0$, let $b\in\bbD_0$ with $Rb\cong Lc'$ and one has
\begin{align*}
\bbC(c,c')&\leq\bbX(Lc',Lc')\circ\bbC(c,c')\\
&=\bbX(Lc',Rb)\circ\bbC(c,c')&(Rb\cong Lc')\\
&=\bbC(c',Tb)\circ\bbC(c,c')&(T^{\nat}=S_{\nat}=\bbX(L-,R-))\\
&\leq\bbC(c,Tb)\\
&=\bbX(Lc,Rb)&(T^{\nat}=S_{\nat}=\bbX(L-,R-))\\
&=\bbX(Lc,Lc'),&(Rb\cong Lc')
\end{align*}
showing that $L$ is a $\CQ$-functor, and the $\CQ$-functoriality of $R$ can be proved similarly.
\end{proof}

It is readily seen that Corollary \ref{general_representation_type_preserving} reduces to Theorem \ref{general_representation_poset} when $\CQ={\bf 2}$. However, in general the $\CQ$-categories $\bbC$ and $\bbD$ may be too ``large'' to compute whether a $\CQ$-category $\bbX$ is equivalent to $\Fix(TS)$, and one would wish to find $\CQ$-categories with smaller size than $\bbC$ and $\bbD$ which are able to generate the required $\CQ$-functors $L:\bbC\lra\bbX$ and $R:\bbD\lra\bbX$. A natural way is through dense and codense $\CQ$-functors introduced in the next section.

\section{Weighted (co)limits, Kan extensions and (co)dense $\CQ$-functors} \label{Colimits_Kan_codense}

\subsection{Weighted (co)limits in $\CQ$-categories}

For each $q\in\CQ_0$, Let $\{q\}$ denote the discrete $\CQ$-category with only one object $q$ such that $|q|=q$ and $\{q\}(q,q)=1_q$. A \emph{presheaf} with type $q$ on a $\CQ$-category $\bbA$ is a $\CQ$-distributor $\mu:\bbA\oto\{q\}$. Presheaves on $\bbA$ constitute a $\CQ$-category $\PA$ with
$$\PA(\mu,\mu')=\mu'\lda\mu$$
for all $\mu,\mu'\in\PA$. Dually, the $\CQ$-category $\PdA$ of \emph{copresheaves} on $\bbA$ consists of $\CQ$-distributors $\lam:\{q\}\oto\bbA$ as objects with type $q$ and $$\PdA(\lam,\lam')=\lam'\rda\lam$$
for all $\lam,\lam'\in\PdA$. It is easy to see $\PdA\cong(\PA^{\op})^{\op}$ as remarked in \ref{Qcat_dual}.

\begin{rem} \label{PdA_QDist_order}
The underlying order of $\PdA$ is precisely the \emph{reverse} local order in $\QDist$; that is, $\mu\leq\lam$ in the underlying order of $\PdA$ if and only if $\lam\leq\mu$ in $\QDist$. In order to avoid confusion, we make the convention that the symbol $\leq$ between $\CQ$-distributors always refer to the local order in $\QDist$. Moreover, while $\bv$ and $\bw$ are used as generic symbols for joins and meets, we write $\bigsqcup$ and $\bigsqcap$ instead for the underlying joins and meets in $\PdA$ to eliminate ambiguity.
\end{rem}

Given a $\CQ$-category $\bbA$, the \emph{Yoneda embedding} $\sY_{\bbA}:\bbA\lra\PA$ sends each $x\in\bbA_0$ to $\bbA(-,x)\in\PA$, and the \emph{co-Yoneda embedding} $\sYd_{\bbA}:\bbA\lra\PdA$ sends each $x\in\bbA_0$ to $\bbA(x,-)\in\PdA$. Both $\sY_{\bbA}$ and $\sYd_{\bbA}$ are fully faithful $\CQ$-functors as the following Yoneda lemma implies:

\begin{lem}[Yoneda] (See \cite{Stubbe2005}.) \label{Yoneda_lemma}
Let $\bbA$ be a $\CQ$-category and $\mu\in\PA$, $\lam\in\PdA$. Then
$$\mu=\PA(\sY_{\bbA}-,\mu)=(\sY_{\bbA})_{\nat}(-,\mu),\quad\lam=\PdA(\lam,\sYd_{\bbA}-)=(\sYd_{\bbA})^{\nat}(\lam,-).$$
\end{lem}

Given a $\CQ$-functor $F:\bbX\lra\bbA$, the \emph{colimit} of $F$ weighted by a presheaf $\mu\in\PX$ is an object $\colim_{\mu}F\in\bbA_0$ of type $|\mu|$ such that
\begin{equation} \label{colim_def}
\bbA({\colim}_{\mu}F,-)=F_{\nat}\lda\mu.
\end{equation}
In particular, $\sup_{\bbA}\mu:={\colim}_{\mu}1_{\bbA}$, when it exists, is called the \emph{supremum} of $\mu\in\PA$, which satisfies
\begin{equation} \label{sup_def}
\bbA({\sup}_{\bbA}\mu,-)=\bbA\lda\mu.
\end{equation}
Dually, the \emph{limit} of $F:\bbX\lra\bbA$ weighted by a copresheaf $\lam\in\PdX$ is defined as $\lim_{\lam}F=\colim_{\lam^{\op}}F^{\op}$; that is, an object $\lim_{\lam}F\in\bbA_0$ of type $|\lam|$ such that
\begin{equation} \label{lim_def}
\bbA(-,{\lim}_{\lam}F)=\lam\rda F^{\nat}.
\end{equation}
The \emph{infimum} of $\lam\in\PdA$, when it exists, is given by $\inf_{\bbA}\lam:=\lim_{\lam}1_{\bbA}$.

\begin{prop} \label{colim_as_sup} (See \cite{Stubbe2005}.)
For all $\CQ$-functors $F:\bbX\lra\bbA$ and $\mu\in\PX$, $\lam\in\PdX$,
$${\colim}_{\mu}F={\sup}_{\bbA}F^{\ra}\mu\quad\text{and}\quad{\lim}_{\lam}F={\inf}_{\bbA}F^{\nra}\lam,$$
where the $\CQ$-functors $F^{\ra}:\PX\lra\PA$ and $F^{\nra}:\PdX\lra\PdA$ are given by
$$F^{\ra}\mu=\mu\circ F^{\nat}\quad\text{and}\quad F^{\nra}\lam=F_{\nat}\circ\lam.$$
\end{prop}

A $\CQ$-category $\bbA$ is \emph{complete} if it satisfies one of the equivalent conditions in the following theorem. In particular, $\PA$ and $\PdA$ are both separated complete $\CQ$-categories.

\begin{thm} (See \cite{Stubbe2005}.) \label{QCat_complete_equivalent}
For any $\CQ$-category $\bbA$, the following conditions are equivalent:
\begin{enumerate}[\rm (i)]
\item $\bbA$ admits all weighted colimits.
\item $\bbA$ admits all weighted limits.
\item Every $\mu\in\PA$ has a supremum.
\item Every $\lam\in\PdA$ has an infimum.
\item $\sY_{\bbA}$ has a left adjoint $\sup_{\bbA}:\PA\lra\bbA$ in $\QCat$.
\item $\sYd_{\bbA}$ has a right adjoint $\inf_{\bbA}:\PdA\lra\bbA$ in $\QCat$.
\end{enumerate}
\end{thm}

It is well known that fixed points of a $\CQ$-closure operator or $\CQ$-interior operator on a complete $\CQ$-category constitute a complete $\CQ$-category:

\begin{prop} \label{AF_complete} (See \cite{Shen2013a}.)
Let $F:\bbA\lra\bbA$ be a $\CQ$-closure operator (resp. $\CQ$-interior operator) on a complete $\CQ$-category $\bbA$. Then $\Fix(F)$ is also a complete $\CQ$-category.
\end{prop}

\subsection{Kan extensions of $\CQ$-functors} \label{Kan_extensions}

Given $\CQ$-functors $K:\bbA\lra\bbB$ and $F:\bbA\lra\bbC$, the (pointwise) \emph{left Kan extension} \cite{Stubbe2005} of $F$ along $K$, when it exists, is given by
\begin{equation} \label{left_Kan_def}
\Lan_K F:\bbB\lra\bbC,\quad(\Lan_K F)b={\colim}_{K_{\nat}(-,b)}F.
\end{equation}

\begin{rem}
The (non-pointwise) left Kan extension of $F:\bbA\lra\bbC$ along $K:\bbA\lra\bbB$, when it exists, is a $\CQ$-functor $\Lan_K F:\bbB\lra\bbC$ with
\begin{equation} \label{left_Kan_ineq}
\Lan_K F\leq S\iff F\leq SK
\end{equation}
for all $\CQ$-functors $S:\bbB\lra\bbC$. It is easy to see that pointwise left Kan extensions defined by \eqref{left_Kan_def} always satisfy \eqref{left_Kan_ineq}, but not vice versa. All Kan extensions considered in this paper are pointwise.
\end{rem}

Dually, the (pointwise) \emph{right Kan extension} of $F$ along $K$ is given by
\begin{equation} \label{right_Kan_def}
\Ran_K F=(\Lan_{K^{\op}}F^{\op})^{\op}:\bbB\lra\bbC,\quad(\Ran_K F)b={\lim}_{K^{\nat}(b,-)}F.
\end{equation}

From Equations \eqref{colim_def} and \eqref{lim_def} one soon has the following characterization of Kan extensions:

\begin{prop} \label{Kan_graph}
$G:\bbB\lra\bbC$ is the left (resp. right) Kan extension of $F:\bbA\lra\bbC$ along $K:\bbA\lra\bbB$ if, and only if,
$$G_{\nat}=F_{\nat}\lda K_{\nat}\quad(\text{resp.}\ G^{\nat}=K^{\nat}\rda F^{\nat}).$$
\end{prop}

From Proposition \ref{Kan_graph} one may derive several useful formulas regarding to Kan extensions:

\begin{prop} \phantomsection \label{Kan_formula}
\begin{enumerate}[\rm (1)]
\item For any $\CQ$-functor $F:\bbA\lra\bbB$, $F\cong\Lan_{1_{\bbA}}F\cong\Ran_{1_{\bbA}}F$.
\item For $\CQ$-functors $F:\bbA\lra\bbC$, $F':\bbA\lra\bbC'$, $G:\bbB\lra\bbC$, $G':\bbB\lra\bbC'$, $K:\bbA\lra\bbX$, $H:\bbB\lra\bbY$,
$$(\Ran_H G)^{\nat}\circ(\Lan_K F)_{\nat}=(\Ran_H G')^{\nat}\circ(\Lan_K F')_{\nat}$$
whenever $G^{\nat}\circ F_{\nat}=G'^{\nat}\circ F'_{\nat}$ and $\Lan_K F$, $\Lan_K F'$, $\Ran_H G$ and $\Ran_H G'$ exist.
\end{enumerate}
\end{prop}

\begin{proof}
(1) is trivial. For (2), note that
\begin{align*}
(\Ran_H G)^{\nat}\circ(\Lan_K F)_{\nat}&=(H^{\nat}\rda G^{\nat})\circ(\Lan_K F)_{\nat}&(\text{Proposition \ref{Kan_graph}})\\
&=H^{\nat}\rda(G^{\nat}\circ(\Lan_K F)_{\nat})&(\text{Proposition \ref{adjoint_arrow_calculation}(3)})\\
&=H^{\nat}\rda(G^{\nat}\circ(F_{\nat}\lda K_{\nat}))&(\text{Proposition \ref{Kan_graph}})\\
&=H^{\nat}\rda((G^{\nat}\circ F_{\nat})\lda K_{\nat}),&(\text{Proposition \ref{adjoint_arrow_calculation}(3)})
\end{align*}
and similarly one has $(\Ran_H G')^{\nat}\circ(\Lan_K F')_{\nat}=H^{\nat}\rda((G'^{\nat}\circ F'_{\nat})\lda K_{\nat})$. Thus $G^{\nat}\circ F_{\nat}=G'^{\nat}\circ F'_{\nat}$ implies $(\Ran_H G)^{\nat}\circ(\Lan_K F)_{\nat}=(\Ran_H G')^{\nat}\circ(\Lan_K F')_{\nat}$.
\end{proof}

The identity in Proposition \ref{Kan_formula}(2) may be translated through Proposition \ref{distributor_graph} as $\bbC(F-,G-)=\bbC'(F'-,G'-)$ implying $\bbC((\Lan_K F)-,(\Ran_H G)-)=\bbC'((\Lan_K F')-,(\Ran_H G')-)$ as the following diagram illustrates:
$$\bfig
\morphism(-1500,0)|b|<800,0>[\bbA`\bbX;K]
\morphism(-1500,0)|a|<1500,500>[\bbA`\bbC;F]
\morphism(-700,0)|b|<700,500>[\bbX`\bbC;\Lan_K F]
\morphism(-1500,0)|b|<1500,-500>[\bbA`\bbC';F']
\morphism(-700,0)|a|<700,-500>[\bbX`\bbC';\Lan_K F']
\morphism(1500,0)|b|<-800,0>[\bbB`\bbY;H]
\morphism(1500,0)|a|<-1500,500>[\bbB`\bbC;G]
\morphism(700,0)|b|<-700,500>[\bbY`\bbC;\Ran_H G]
\morphism(1500,0)|b|<-1500,-500>[\bbB`\bbC';G']
\morphism(700,0)|a|<-700,-500>[\bbY`\bbC';\Ran_H G']
\efig$$

If $\Lan_K F:\bbB\lra\bbC$ (resp. $\Ran_K F:\bbB\lra\bbC$) exists, a $\CQ$-functor $H:\bbC\lra\bbD$ is said to \emph{preserve} $\Lan_K F$ (resp. $\Ran_K F$) if $\Lan_K HF$ (resp. $\Ran_K HF$) exists and is isomorphic to $H\Lan_K F$ (resp. $H\Ran_K F$). $\Lan_K F$ (resp. $\Ran_K F$) is \emph{absolute} if it is preserved by any $\CQ$-functor with domain $\bbC$. The following characterization of adjoint $\CQ$-functors appeared in \cite{Stubbe2005} in terms of non-pointwise Kan extensions, and here we strengthen it to the pointwise version:

\begin{prop} \label{Kan_extension_adjoint}
Let $F:\bbA\lra\bbB$ be a $\CQ$-functor. The following statements are equivalent:
\begin{enumerate}[\rm (i)]
\item $F$ is a left (resp. right) adjoint in $\QCat$.
\item $\Lan_F 1_{\bbA}$ (resp. $\Ran_F 1_{\bbA}$) exists and is absolute.
\item $\Lan_F 1_{\bbA}$ (resp. $\Ran_F 1_{\bbA}$) exists and is preserved by $F$.
\end{enumerate}
In this case, $\Lan_F 1_{\bbA}:\bbB\lra\bbA$ (resp. $\Ran_F 1_{\bbA}:\bbB\lra\bbA$) is the right (resp. left) adjoint of $F$.
\end{prop}

\begin{proof}
(i)${}\Lra{}$(ii): If $F\dv G$, for the existence of $\Lan_F 1_{\bbA}$ it suffices to prove $G\cong\Lan_F 1_{\bbA}$. Indeed, from Propositions \ref{adjoint_arrow_calculation}(1) and \ref{adjoint_graph} one has
$$G_{\nat}=\bbA\lda G^{\nat}=(1_{\bbA})_{\nat}\lda F_{\nat},$$
and thus Proposition \ref{Kan_graph} guarantees $G\cong\Lan_F 1_{\bbA}$. Now let $H:\bbA\lra\bbC$ be any $\CQ$-functor, by applying again Propositions \ref{adjoint_arrow_calculation}(1) and \ref{adjoint_graph} one has
$$(HG)_{\nat}=H_{\nat}\circ G_{\nat}=H_{\nat}\lda G^{\nat}=H_{\nat}\lda F_{\nat},$$
showing that $HG\cong\Lan_F H$.

(ii)${}\Lra{}$(iii): Trivial.

(iii)${}\Lra{}$(i): Let $G=\Lan_F 1_{\bbA}$, then $FG\cong\Lan_F F$. By Proposition \ref{Kan_graph} one has
$$G_{\nat}=\bbA\lda F_{\nat}\quad\text{and}\quad F_{\nat}\lda F_{\nat}=(FG)_{\nat}=F_{\nat}\circ G_{\nat}=F_{\nat}\lda G^{\nat},$$
where the last equality follows from Proposition \ref{adjoint_arrow_calculation}(1). It follows that
$$F_{\nat}\leq G_{\nat}\rda\bbA=G^{\nat}\leq(F_{\nat}\lda F_{\nat})\rda F_{\nat}=F_{\nat},$$
where the first equality follows from Proposition \ref{adjoint_arrow_calculation}(1). Thus $F_{\nat}=G^{\nat}$ and by Proposition \ref{adjoint_graph} one has $F\dv G$.
\end{proof}

%

The following characterizations of adjoint $\CQ$-functors will be useful in the sequel:

\begin{prop} \label{la_condition} (See \cite{Stubbe2005}.)
Let $F:\bbA\lra\bbB$ be a $\CQ$-functor. If $F$ is a left (resp. right) adjoint in $\QCat$, then
\begin{enumerate}[\rm (1)]
\item $F$ is \emph{cocontinuous} (resp. \emph{continuous}) in the sense that $F\colim_{\mu}G\cong\colim_{\mu}FG$ (resp. $F\lim_{\lam}G\cong\lim_{\lam}FG$) for all $\CQ$-functors $G:\bbX\lra\bbA$ and $\mu\in\PX$ (resp. $\lam\in\PdX$).
\item $F$ is \emph{$\sup$-preserving} (resp. \emph{$\inf$-preserving}) in the sense that $F\sup_{\bbA}\cong\sup_{\bbB}F^{\ra}$ (resp. $F\inf_{\bbA}\cong\inf_{\bbB}F^{\nra}$).
\item $F$ preserves left (resp. right) Kan extensions of any $\CQ$-functor with codomain $\bbA$.
\item $F$ is a left (resp. right) adjoint between the underlying ordered sets of $\bbA$, $\bbB$.
\end{enumerate}
Moreover, if $\bbA$ is complete, then the following statements are equivalent:
\begin{enumerate}[\rm (i)]
\item $F$ is a left (resp. right) adjoint in $\QCat$.
\item $F$ is cocontinuous (resp. continuous).
\item $F$ is $\sup$-preserving (resp. $\inf$-preserving).
\item $F$ preserves left (resp. right) Kan extensions.
\end{enumerate}
\end{prop}

\subsection{(Co)dense $\CQ$-functors}

A $\CQ$-functor $F:\bbA\lra\bbB$ is \emph{dense} \cite{Shen2013a} if for any $y\in\bbB_0$, there exists $\mu\in\PA$ such that $y\cong\colim_{\mu}F$. Dually, $F$ is \emph{codense} if $F^{\op}$ is a dense $\CQ^{\op}$-functor; that is, $y\cong\lim_{\lam}F$ for some $\lam\in\PdA$ for any $y\in\bbB_0$.

A $\CQ$-subcategory $\bbB$ of $\bbA$ is \emph{dense} (resp. \emph{codense}) if the inclusion $\CQ$-functor $J:\bbB\ \to/^(->/\bbA$ is dense (resp. codense).

\begin{exmp} \label{Yoneda_dense}
For any $\CQ$-category $\bbA$, the Yoneda embedding $\sY_{\bbA}:\bbA\lra\PA$ is dense since $\mu=\colim_{\mu}\sY_{\bbA}$ for all $\mu\in\PA$. Dually, the co-Yoneda embedding $\sYd_{\bbA}:\bbA\lra\PdA$ is codense.
\end{exmp}

We have the following equivalent characterizations of dense and codense $\CQ$-functors:

\begin{prop} \label{dense_condition}
Let $F:\bbA\lra\bbB$ be a $\CQ$-functor. The following statements are equivalent:
\begin{enumerate}[\rm (i)]
\item $F$ is dense (resp. codense).
\item $F$ is \emph{$\sup$-dense} (resp. \emph{$\inf$-dense}) in the sense that there exists $\mu\in\PA$ (resp. $\lam\in\PdA$) with $y\cong\sup_{\bbB}F^{\ra}\mu$ (resp. $y\cong\inf_{\bbB}F^{\nra}\lam$) for all $y\in\bbB_0$.
\item $\sIm(F)=\{Fx\mid x\in\bbA_0\}$ is a dense (resp. codense) $\CQ$-subcategory of $\bbB$.
\item $F_{\nat}\lda F_{\nat}=\bbB$ (resp. $F^{\nat}\rda F^{\nat}=\bbB$).
\item $1_{\bbB}\cong\Lan_F F$ (resp. $1_{\bbB}\cong\Ran_F F$).
\end{enumerate}
\end{prop}

\begin{proof}
(i)$\iff$(ii): Follows immediately from Proposition \ref{colim_as_sup}.

(ii)$\iff$(iii): Since one already has (i)$\iff$(ii), it suffices to show that $F$ is $\sup$-dense if, and only if, the inclusion $\CQ$-functor $J:\sIm(F)\ \to/^(->/\bbB$ is $\sup$-dense. Let $G:\bbA\lra\sIm(F)$ be the codomain restriction of $F$, then obviously $F=JG$, and the surjectivity of $G$ implies $\lam\circ G_{\nat}\circ G^{\nat}=\lam$ for all $\lam\in\sP(\sIm(F))$ by Proposition \ref{functor_graph}(2).

Let $y\in\bbB_0$. On one hand, if $y\cong\sup_{\bbB}F^{\ra}\mu$ for some $\mu\in\PA$, then $G^{\ra}\mu\in\sP(\sIm(F))$ satisfies
$$y\cong{\sup}_{\bbB}F^{\ra}\mu={\sup}_{\bbB}J^{\ra}G^{\ra}\mu.$$
On the other hand, if $y\cong\sup_{\bbB}J^{\ra}\lam$ for some $\lam\in\sP(\sIm(F))$, then $\lam\circ G_{\nat}\in\PA$ satisfies
$$y\cong{\sup}_{\bbB}J^{\ra}\lam={\sup}_{\bbB}J^{\ra}(\lam\circ G_{\nat}\circ G^{\nat})={\sup}_{\bbB}J^{\ra}G^{\ra}(\lam\circ G_{\nat})={\sup}_{\bbB}F^{\ra}(\lam\circ G_{\nat}).$$

(i)${}\Lra{}$(iv): One may find $\mu\in\PA$ such that $y\cong\colim_{\mu}F$ for any $y\in\bbB_0$. Then
\begin{align*}
\bbB(y,-)&\leq F_{\nat}\lda F_{\nat}(-,y)\\
&\leq(F_{\nat}\lda F_{\nat}(-,y))\circ\bbB(y,y)\\
&=(F_{\nat}\lda F_{\nat}(-,y))\circ(F_{\nat}(-,y)\lda\mu)&(\text{Equation \eqref{colim_def}})\\
&\leq F_{\nat}\lda\mu\\
&=\bbB(y,-),&(\text{Equation \eqref{colim_def}})
\end{align*}
and consequently $\bbB(y,-)=F_{\nat}\lda F_{\nat}(-,y)=(F_{\nat}\lda F_{\nat})(y,-)$.

(iv)${}\Lra{}$(i): $F_{\nat}\lda F_{\nat}=\bbB$ immediately implies $\bbB(y,-)=F_{\nat}\lda F_{\nat}(-,y)$; that is, $y\cong\colim_{F_{\nat}(-,y)}F$ for any $y\in\bbB_0$.

(iv)$\iff$(v): Follows immediately from Proposition \ref{Kan_graph}.
\end{proof}

\begin{cor} \phantomsection \label{composition_dense}
\begin{enumerate}[\rm (1)]
\item Every essentially surjective $\CQ$-functor is both dense and codense.
\item If $\CQ$-functors $F:\bbA\lra\bbB$ and $G:\bbB\lra\bbC$ are both dense (resp. codense) and $G$ is a left (resp. right) adjoint in $\QCat$, then $GF:\bbA\lra\bbC$ is dense (resp. codense).
\end{enumerate}
\end{cor}

\begin{proof}
(1) is easy. For (2), note that $G_{\nat}=H^{\nat}$ if $G\dv H$ in $\QCat$, and thus
\begin{align*}
\bbC&=G_{\nat}\lda G_{\nat}&(\text{Proposition \ref{dense_condition}(iv)})\\
&=G_{\nat}\circ(\bbB\lda G_{\nat})&(G_{\nat}=H^{\nat}\ \text{and Proposition \ref{adjoint_arrow_calculation}(3)})\\
&=G_{\nat}\circ((F_{\nat}\lda F_{\nat})\lda G_{\nat})&(\text{Proposition \ref{dense_condition}(iv)})\\
&=G_{\nat}\circ(F_{\nat}\lda(GF)_{\nat})\\
&=(GF)_{\nat}\lda(GF)_{\nat},&(G_{\nat}=H^{\nat}\ \text{and Proposition \ref{adjoint_arrow_calculation}(3)})
\end{align*}
showing that $GF$ is dense.
\end{proof}

\section{Representation theorem in terms of (co)dense $\CQ$-functors} \label{General_Representation_Dense}

Now we are ready to present the second main result of this paper. If $S\dv T:\bbD\lra\bbC$ is an adjunction between complete $\CQ$-categories, Proposition \ref{AF_complete} guarantees the completeness of $\Fix(TS)\simeq\Fix(ST)$. In this case, the following representation theorem can be established through dense and codense $\CQ$-functors:

\begin{thm} \label{general_representation_dense_complete}
Let $S\dv T:\bbD\lra\bbC$ be an adjunction between complete $\CQ$-categories. Then a complete $\CQ$-category $\bbX$ is equivalent to $\Fix(TS)$ if, and only if, there exist dense $\CQ$-functors $F:\bbA\lra\bbX$, $K:\bbA\lra\bbC$ and codense $\CQ$-functors $G:\bbB\lra\bbX$, $H:\bbB\lra\bbD$ with $H^{\nat}\circ S_{\nat}\circ K_{\nat}=G^{\nat}\circ F_{\nat}$.
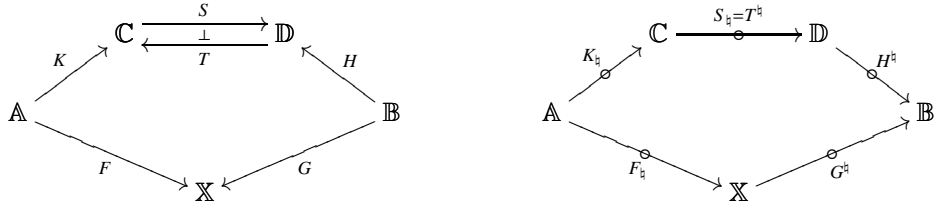
$$\bfig
\morphism(-700,0)<400,300>[\bbA`\bbC;K]
\morphism(700,0)<-400,300>[\bbB`\bbD;H]
\morphism(-700,0)|b|<700,-300>[\bbA`\bbX;F]
\morphism(700,0)|b|<-700,-300>[\bbB`\bbX;G]
\morphism(-300,300)|a|/@{>}@<4pt>/<600,0>[\bbC`\bbD;S]
\morphism(300,300)|b|/@{>}@<4pt>/<-600,0>[\bbD`\bbC;T]
\place(0,305)[\mbox{\scriptsize$\bot$}]
\morphism(1300,0)<400,300>[\bbA`\bbC;K_{\nat}]
\morphism(1700,300)|a|<600,0>[\bbC`\bbD;S_{\nat}=T^{\nat}]
\morphism(2300,300)<400,-300>[\bbD`\bbB;H^{\nat}]
\morphism(1300,0)|b|<700,-300>[\bbA`\bbX;F_{\nat}]
\morphism(2000,-300)|b|<700,300>[\bbX`\bbB;G^{\nat}]
\place(1500,150)[\circ] \place(2000,300)[\circ] \place(2500,150)[\circ] \place(1650,-150)[\circ] \place(2350,-150)[\circ]
\efig$$
\end{thm}

\begin{proof}
{\bf Necessity.} One may find essentially surjective $\CQ$-functors $L:\bbC\lra\bbX$ and $R:\bbD\lra\bbX$ with $S_{\nat}=R^{\nat}\circ L_{\nat}$ by Theorem \ref{general_representation}. Then the dense $\CQ$-functors $L:\bbC\lra\bbX$, $1_{\bbC}:\bbC\lra\bbC$ and the codense $\CQ$-functors $R:\bbD\lra\bbX$, $1_{\bbD}:\bbD\lra\bbD$ clearly satisfy $R^{\nat}\circ L_{\nat}=S_{\nat}=1_{\bbD}^{\nat}\circ S_{\nat}\circ(1_{\bbC})_{\nat}$.

{\bf Sufficiency.} The completeness of $\bbX$ guarantees the existence of the Kan extensions (see the definitions in \eqref{left_Kan_def} and \eqref{right_Kan_def})
$$L:=\Lan_K F:\bbC\lra\bbX\quad\text{and}\quad R:=\Ran_H G:\bbD\lra\bbX.$$
\begin{equation} \label{KFL_ST_HGR}
\bfig
\morphism(-1100,0)<800,300>[\bbA`\bbC;K]
\morphism(1100,0)<-800,300>[\bbB`\bbD;H]
\morphism(-1100,0)|b|<1100,-300>[\bbA`\bbX;F]
\morphism(1100,0)|b|<-1100,-300>[\bbB`\bbX;G]
\morphism(-300,300)|a|/@{>}@<4pt>/<600,0>[\bbC`\bbD;S]
\morphism(300,300)|b|/@{>}@<4pt>/<-600,0>[\bbD`\bbC;T]
\morphism(-300,300)|l|<300,-600>[\bbC`\bbX;L=\Lan_K F]
\morphism(300,300)|r|<-300,-600>[\bbD`\bbX;R=\Ran_H G]
\place(0,305)[\mbox{\scriptsize$\bot$}]
\efig
\end{equation}
We show that $L$ and $R$ satisfy the conditions in Theorem \ref{general_representation}.

First, $LK\cong F$ and $RH\cong G$; that is, $L$ and $R$ are actually extensions of $F$ and $G$, respectively. For this, note that
\begin{align*}
F_{\nat}&=(\Ran_G G)^{\nat}\circ(\Lan_{1_{\bbA}}F)_{\nat}&(\text{Propositions \ref{Kan_formula}(1) and \ref{dense_condition}(v)})\\
&=(\Ran_G H)^{\nat}\circ(\Lan_{1_{\bbA}}SK)_{\nat}&(H^{\nat}\circ S_{\nat}\circ K_{\nat}=G^{\nat}\circ F_{\nat}\ \text{and Proposition \ref{Kan_formula}(2)})\\
&=(\Ran_G H)^{\nat}\circ(SK)_{\nat}&(\text{Propositions \ref{Kan_formula}(1)})\\
&=(\Ran_G H)^{\nat}\circ(S\Lan_K K)_{\nat}\circ K_{\nat}&(\text{Proposition \ref{dense_condition}(v)})\\
&=(\Ran_G H)^{\nat}\circ(\Lan_K SK)_{\nat}\circ K_{\nat}&(\text{Proposition \ref{la_condition}(3)})\\
&=(\Ran_G G)^{\nat}\circ(\Lan_K F)_{\nat}\circ K_{\nat}&(H^{\nat}\circ S_{\nat}\circ K_{\nat}=G^{\nat}\circ F_{\nat}\ \text{and Proposition \ref{Kan_formula}(2)})\\
&=L_{\nat}\circ K_{\nat},&(\text{Proposition \ref{dense_condition}(v)})
\end{align*}
and thus the conclusion follows (see \eqref{F_leq_G_graph}). Similarly one may prove $RH\cong G$.

Second, $L$ and $R$ are essentially surjective. To this end, note that $S_{\nat}=T^{\nat}$ implies
\begin{equation} \label{GFnat=THKnat}
G^{\nat}\circ F_{\nat}=H^{\nat}\circ S_{\nat}\circ K_{\nat}=H^{\nat}\circ T^{\nat}\circ K_{\nat}=(TH)^{\nat}\circ K_{\nat},
\end{equation}
and consequently
\begin{align*}
L_{\nat}&=(\Ran_G G)^{\nat}\circ(\Lan_K F)_{\nat}&(\text{Proposition \ref{dense_condition}(v)})\\
&=(\Ran_G TH)^{\nat}\circ(\Lan_K K)_{\nat}&(\text{Equation \eqref{GFnat=THKnat} and Proposition \ref{Kan_formula}(2)})\\
&=(\Ran_G TH)^{\nat},&(\text{Proposition \ref{dense_condition}(v)})
\end{align*}
where the existence of $\Ran_G TH$ is guaranteed by the completeness of $\bbC$. Thus $L\dv\Ran_G TH$ in $\QCat$ and, as a left adjoint, $L$ is cocontinuous (see Proposition \ref{la_condition}(1)). Therefore, for any $x\in\bbX_0$ one may find $\mu\in\PA$ with $x\cong\colim_{\mu}F$, and consequently ${\colim}_{\mu}K\in\bbC_0$ satisfies
$$L({\colim}_{\mu}K)\cong{\colim}_{\mu}LK\cong{\colim}_{\mu}F\cong x,$$
where the second isomorphism follows from $LK\cong F$. Hence $L$ is essentially surjective. Similarly one may obtain the essential surjectivity of $R$ by showing that $R$ is a right adjoint in $\QCat$ (with $\Lan_F SK$ as its left adjoint) and applying $RH\cong G$.

Finally, $S_{\nat}=R^{\nat}\circ L_{\nat}$. Indeed,
\begin{align*}
S_{\nat}&=(S\Lan_K K)_{\nat}&(\text{Proposition \ref{dense_condition}(v)})\\
&=(\Lan_K SK)_{\nat}&(\text{Proposition \ref{la_condition}(3)})\\
&=(\Ran_H H)^{\nat}\circ(\Lan_K SK)_{\nat}&(\text{Proposition \ref{dense_condition}(v)})\\
&=(\Ran_H G)^{\nat}\circ(\Lan_K F)_{\nat}&(H^{\nat}\circ S_{\nat}\circ K_{\nat}=G^{\nat}\circ F_{\nat}\ \text{and Proposition \ref{Kan_formula}(2)})\\
&=R^{\nat}\circ L_{\nat},
\end{align*}
which completes the proof.
\end{proof}

Theorem \ref{general_representation_dense_complete} paves a way towards a ``good'' representation of $\Fix(TS)\simeq\Fix(ST)$ for a specific adjunction $S\dv T:\bbD\lra\bbC$ between complete $\CQ$-categories; that is, by looking for dense $\CQ$-functors into $\bbC$ (or equivalently, dense $\CQ$-subcategories of $\bbC$) and codense $\CQ$-functors into $\bbD$ (or equivalently, codense $\CQ$-subcategories of $\bbD$). The power of this theorem will be revealed in the next section for representations of concept lattices.

\section{Fixed points of Isbell adjunctions and Kan adjunctions} \label{Isbell_Kan_fixed_points}

In this section we demonstrate how the general representation theorems (\ref{general_representation} and \ref{general_representation_dense_complete}) give rise to representation theorems of concept lattices in FCA and RST in the generality of the $\CQ$-version.

\subsection{Isbell adjunctions and Kan adjunctions}

Each $\CQ$-distributor $\phi:\bbA\oto\bbB$ induces an \emph{Isbell adjunction} $\uphi\dv\dphi:\PdB\lra\PA$ in $\QCat$ \cite{Shen2013a} given by
\begin{align*}
&\uphi:\PA\lra\PdB,\quad\mu\mapsto\phi\lda\mu,\\
&\dphi:\PdB\lra\PA,\quad\lam\mapsto\lam\rda\phi
\end{align*}
and a \emph{Kan adjunction} $\phi^*\dv\phi_*:\PA\lra\PB$ defined as
\begin{align*}
&\phi^*:\PB\lra\PA,\quad \lam\mapsto\lam\circ\phi,\\
&\phi_*:\PA\lra\PB,\quad \mu\mapsto\mu\lda\phi.
\end{align*}
Since $\PdA\cong(\PA^{\op})^{\op}$, there is also a \emph{dual Kan adjunction} $\phi_{\dag}\dv\phi^{\dag}$ with
\begin{align*}
&\phi_{\dag}:=((\phi^{\op})_*)^{\op}:\PdB\lra\PdA,\quad\lam\mapsto\phi\rda\lam,\\
&\phi^{\dag}:=((\phi^{\op})^*)^{\op}:\PdA\lra\PdB,\quad\mu\mapsto\phi\circ\mu
\end{align*}
which corresponds to the Kan adjunction $(\phi^{\op})^*\dv(\phi^{\op})_*:\PB^{\op}\lra\PA^{\op}$ in $\CQ^{\op}\text{-}\Cat$ under the isomorphism \eqref{QopCat_iso} in Remark \ref{Qcat_dual}.

\begin{rem}
As left and right Kan extensions of $\CQ$-functors introduced in Subsection \ref{Kan_extensions} are exactly \emph{left} and \emph{right extensions} of 1-cells \cite{Lack2010} in the 2-category $\QCat$, Kan adjunctions induced by $\CQ$-distributors in fact generalize \emph{right extensions} of 1-cells in the 2-category $\QDist$. To see this, note that the underlying Galois connection (between the underlying ordered sets) of a Kan adjunction $\phi^*\dv\phi_*:\PA\lra\PB$ may be described as the monotone map ``composing with $\phi$''
$$\phi^*:\QDist(\bbB,\{q\})\lra\QDist(\bbA,\{q\})$$
admitting a right adjoint
$$\phi_*:\QDist(\bbA,\{q\})\lra\QDist(\bbB,\{q\})$$
for any $q\in\CQ_0$, which exactly says that for any presheaf $\mu:\bbA\oto\{q\}$,
$$\phi_*\mu=\mu\lda\phi:\bbB\oto\{q\}$$
is the right extension of $\mu$ along $\phi:\bbA\oto\bbB$. Similarly, dual Kan adjunctions induced by $\CQ$-distributors correspond to \emph{left extensions} of 1-cells in $\QDist$.
\end{rem}

\begin{prop} (See \cite{Heymans2010}.) \label{star_graph_adjoint}
$(-)^*:(\QDist)^{\op}\lra\QCat$ and $(-)^{\dag}:(\QDist)^{\co}\lra\QCat$ are both 2-functorial, and one has two pairs of adjoint 2-functors
$$(-)^{\nat}\dv(-)^*:(\QDist)^{\op}\lra\QCat\quad\text{and}\quad(-)_{\nat}\dv(-)^{\dag}:(\QDist)^{\co}\lra\QCat.$$
\end{prop}

The adjunctions $(-)^{\nat}\dv(-)^*$ and $(-)_{\nat}\dv(-)^{\dag}$ give rise to isomorphisms
$$(\QCat)^{\co}(\bbA,\PdB)\cong\QDist(\bbA,\bbB)\cong\QCat(\bbB,\PA)$$
for all $\CQ$-categories $\bbA$, $\bbB$. We denote by
\begin{equation} \label{olphi_def}
\olphi:\bbB\lra\PA,\quad\olphi y=\phi(-,y)\quad\text{and}\quad\ophi:\bbA\lra\PdB,\quad\ophi x=\phi(x,-)
\end{equation}
for the transposes of each $\CQ$-distributor $\phi:\bbA\oto\bbB$.

\begin{prop} \label{olphi_ophi_Yoneda} (See \cite{Shen2016,Shen2013a}.)
Let $\phi:\bbA\oto\bbB$ be a $\CQ$-distributor. Then
\begin{enumerate}[\rm (1)]
\item $\phi=(\sYd_{\bbB})^{\nat}\circ\ophi_{\nat}=\olphi^{\,\nat}\circ(\sY_{\bbA})_{\nat}$,
\item $\ophi=\uphi\sY_{\bbA}=\phi^{\dag}\sYd_{\bbA}$,
\item $\olphi=\dphi\sYd_{\bbB}=\phi^*\sY_{\bbB}$.
\end{enumerate}
\end{prop}

By Proposition \ref{AF_complete}, fixed points of the Isbell adjunction $\uphi\dv\dphi$ and the Kan adjunction $\phi^*\dv\phi_*$ constitute complete $\CQ$-categories
$$\Mphi:=\Fix(\dphi\uphi)=\{\mu\in\PA\mid\dphi\uphi\mu=\mu\}\quad\text{and}\quad\Kphi:=\Fix(\phi_*\phi^*)=\{\lam\in\PB\mid\phi_*\phi^*\lam=\lam\},$$
where are both separated since so are $\PA$ and $\PB$.

\begin{exmp} \label{MA_KA}
For the identity $\CQ$-distributor $\bbA:\bbA\oto\bbA$ on a $\CQ$-category $\bbA$, $\sM\bbA$ is the MacNeille completion of $\bbA$ \cite{Shen2013a}, and $\sK\bbA=\PA$ is the free cocompletion of $\bbA$ \cite{Stubbe2005}.
\end{exmp}


With the help of Theorems \ref{general_representation} and \ref{general_representation_dense_complete}, we arrive at the following representation theorem of $\Mphi$ as the $\CQ$-version of Theorem \ref{FCA_fundamental}:

\begin{thm} (See \cite{Shen2013a}.) \label{Mphi_representation}
For any $\CQ$-distributor $\phi:\bbA\oto\bbB$, a separated complete $\CQ$-category $\bbX$ is isomorphic to $\Mphi$ if, and only if, there exist a dense $\CQ$-functor $F:\bbA\lra\bbX$ and a codense $\CQ$-functor $G:\bbB\lra\bbX$ with $\phi=G^{\nat}\circ F_{\nat}$.
\end{thm}

\begin{proof}
{\bf Necessity.} By Theorem \ref{general_representation} there exist surjective $\CQ$-functors $L:\PA\lra\bbX$ and $R:\PdB\lra\bbX$ with $(\uphi)_{\nat}=R^{\nat}\circ L_{\nat}$. Since $\sY_{\bbA}:\bbA\lra\PA$ is dense (Example \ref{Yoneda_dense}), $L$ is dense (Corollary \ref{composition_dense}(1)) and $L$ is a left adjoint in $\QCat$ (Corollary \ref{L_la_R_ra}), one deduces the density of $L\sY_{\bbA}:\bbA\lra\bbX$ by Corollary \ref{composition_dense}(2). Similarly one can see that $R\sYd_{\bbB}:\bbB\lra\bbX$ is codense. Finally, one has
$$\phi=(\sYd_{\bbB})^{\nat}\circ\ophi_{\nat}=(\sYd_{\bbB})^{\nat}\circ(\uphi\sY_{\bbA})_{\nat}=(\sYd_{\bbB})^{\nat}\circ R^{\nat}\circ L_{\nat}\circ(\sY_{\bbA})_{\nat}
=(R\sYd_{\bbB})^{\nat}\circ(L\sY_{\bbA})_{\nat}$$
by applying the formulas in Proposition \ref{olphi_ophi_Yoneda}.

{\bf Sufficiency.} Now we have dense $\CQ$-functors $\sY_{\bbA}:\bbA\lra\PA$, $F:\bbA\lra\bbX$ and codense $\CQ$-functors $\sYd_{\bbB}:\bbB\lra\PdB$, $G:\bbB\lra\bbX$ with
$$G^{\nat}\circ F_{\nat}=\phi=(\sYd_{\bbB})^{\nat}\circ\ophi_{\nat}=(\sYd_{\bbB})^{\nat}\circ(\uphi\sY_{\bbA})_{\nat}=(\sYd_{\bbB})^{\nat}\circ(\uphi)_{\nat}\circ(\sY_{\bbA})_{\nat}$$
by Proposition \ref{olphi_ophi_Yoneda}, thus the conditions in Theorem \ref{general_representation_dense_complete} are satisfied. This completes the proof.
\end{proof}

The following diagrams respectively illustrate the ``only if'' part and the ``if'' part of Theorem \ref{Mphi_representation} (cf. Diagram \eqref{KFL_ST_HGR} in the proof of Theorem \ref{general_representation_dense_complete}):
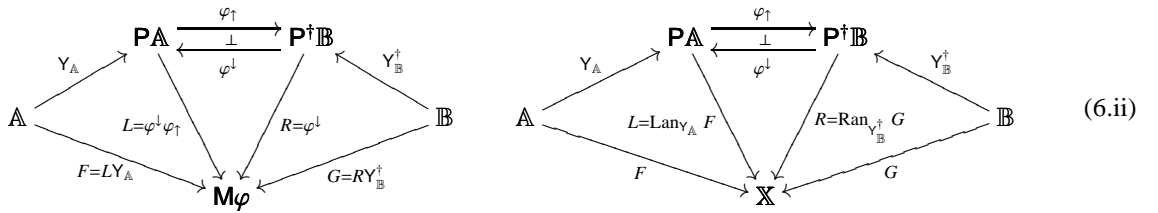
\begin{equation} \label{Mphi_representation_digram}
\bfig
\morphism(-800,0)<500,300>[\bbA`\PA;\sY_{\bbA}]
\morphism(800,0)<-500,300>[\bbB`\PdB;\sYd_{\bbB}]
\morphism(-800,0)|b|<800,-300>[\bbA`\Mphi;F=L\sY_{\bbA}]
\morphism(800,0)|b|<-800,-300>[\bbB`\Mphi;G=R\sYd_{\bbB}]
\morphism(-300,300)|a|/@{>}@<4pt>/<600,0>[\PA`\PdB;\uphi]
\morphism(300,300)|b|/@{>}@<4pt>/<-600,0>[\PdB`\PA;\dphi]
\morphism(-300,300)|l|<300,-600>[\PA`\Mphi;L=\dphi\uphi]
\morphism(300,300)|r|<-300,-600>[\PdB`\Mphi;R=\dphi]
\place(0,305)[\mbox{\scriptsize$\bot$}]
\morphism(1100,0)<600,300>[\bbA`\PA;\sY_{\bbA}]
\morphism(2900,0)<-600,300>[\bbB`\PdB;\sYd_{\bbB}]
\morphism(1100,0)|b|<900,-300>[\bbA`\bbX;F]
\morphism(2900,0)|b|<-900,-300>[\bbB`\bbX;G]
\morphism(1700,300)|a|/@{>}@<4pt>/<600,0>[\PA`\PdB;\uphi]
\morphism(2300,300)|b|/@{>}@<4pt>/<-600,0>[\PdB`\PA;\dphi]
\morphism(1700,300)|l|<300,-600>[\PA`\bbX;L=\Lan_{\sY_{\bbA}}F]
\morphism(2300,300)|r|<-300,-600>[\PdB`\bbX;R=\Ran_{\sYd_{\bbB}}G]
\place(2000,305)[\mbox{\scriptsize$\bot$}]
\efig
\end{equation}

\begin{exmp}
For the identity $\CQ$-distributor $\bbA:\bbA\oto\bbA$ on a $\CQ$-category $\bbA$, $\sM\bbA$ is the MacNeille completion of $\bbA$ (see Example \ref{MA_KA}). In this case, the codomain restriction $\sY_{\bbA}:\bbA\lra\sM\bbA$ of the Yoneda embedding is dense, and the composition $\bbA^{\da}\sYd_{\bbA}:\bbA\lra\PdA\lra\sM\bbA$ is codense. It is not difficult to verify $\bbA=(\bbA^{\da}\sYd_{\bbA})^{\nat}\circ(\sY_{\bbA})_{\nat}$, which fulfills the conditions in Theorem \ref{Mphi_representation}.
\end{exmp}

However, it is not straightforward to derive the counterpart of Theorem \ref{Mphi_representation} for $\Kphi$ as it is not easy to find a non-trivial codense $\CQ$-functor into $\PA$ (or equivalently, a non-trivial codense $\CQ$-subcategory of $\PA$) as required in Theorem \ref{general_representation_dense_complete}. This will be the topic of the next subsection.

\subsection{Towards codense $\CQ$-subcategories of presheaf $\CQ$-categories} \label{Codense_Sub_Presheaf}


In a quantaloid $\CQ$, a family of $\CQ$-arrows $\{d_q:q\lra q\}_{q\in\CQ_0}$ is a \emph{cyclic family} (resp. \emph{dualizing family}) if
\begin{equation} \label{cyclic_dualizing_def}
d_p\lda u=u\rda d_q\quad (\text{resp.}\ (d_p\lda u)\rda d_p=u=d_q\lda(u\rda d_q))
\end{equation}
for all $\CQ$-arrows $u:p\lra q$. A \emph{Girard quantaloid} \cite{Rosenthal1992} is a quantaloid $\CQ$ equipped with a cyclic dualizing family of $\CQ$-arrows. In this case, the \emph{complement} of a $\CQ$-arrow $u:p\lra q$ can be defined as
$$\neg u=d_p\lda u=u\rda d_q:q\lra p,$$
which clearly satisfies $\neg\neg u=u$. For each $\CQ$-category $\bbA$,
$$(\neg\bbA)(y,x)=\neg\bbA(x,y)$$
gives a $\CQ$-distributor $\neg\bbA:\bbA\oto\bbA$, and it is straightforward to check that
$$\{\neg\bbA:\bbA\oto\bbA\}_{\bbA\in\ob(\QDist)}$$
is a cyclic dualizing family of $\QDist$; this gives the ``only if'' part of the following proposition. As for the ``if'' part, just note that $\CQ$ can be fully faithfully embedded in $\QDist$:

\begin{prop} (See \cite{Rosenthal1992}.)
A small quantaloid $\CQ$ is a Girard quantaloid if, and only if, $\QDist$ is a Girard quantaloid.
\end{prop}

Hence, with $\CQ$ being Girard, each $\CQ$-distributor $\phi:\bbA\oto\bbB$ has a \emph{complement}
$$\neg\phi:=\neg\bbA\lda\phi=\phi\rda\neg\bbB:\bbB\oto\bbA.$$

\begin{prop} \label{PA_PdA_iso}
If $\CQ$ is a small Girard quantaloid, then for any $\CQ$-category $\bbA$,
$$\neg:\PA\lra\PdA$$
is an isomorphism in $\QCat$.
\end{prop}

\begin{proof}
Since $\{\neg\bbA\}_{\bbA\in\ob(\QDist)}$ is a cyclic dualizing family, one has
$$\PA(\mu,\lam)=\lam\lda\mu=((\neg\bbA\lda\lam)\rda\neg\bbA)\lda\mu=(\neg\bbA\lda\lam)\rda(\neg\bbA\lda\mu)=\PdA(\neg\mu,\neg\lam)$$
for all $\mu,\lam\in\PA$. Thus $\neg:\PA\lra\PdA$ is a fully faithful $\CQ$-functor, and consequently an isomorphism in $\QCat$ since it is obviously surjective.
\end{proof}

From the combination of Example \ref{Yoneda_dense} and Proposition \ref{PA_PdA_iso} one immediately obtains a codense $\CQ$-functor
\begin{equation} \label{neg_sYd_codense}
\neg\sYd_{\bbA}:=(\bbA\to^{\sYd_{\bbA}}\PdA\to^{\neg}\PA)
\end{equation}
for any $\CQ$-category $\bbA$ when $\CQ$ is Girard. The representation of $\Kphi$ follows easily in this case, which gives the $\CQ$-version of Theorem \ref{RST_fundamental}:

\begin{cor} \label{Kphi_representation_Girard}
Let $\CQ$ be a small Girard quantaloid. Then for any $\CQ$-distributor $\phi:\bbA\oto\bbB$, a separated complete $\CQ$-category $\bbX$ is isomorphic to $\Kphi$ if, and only if, there exist a dense $\CQ$-functor $F:\bbB\lra\bbX$ and a codense $\CQ$-functor $G:\bbA\lra\bbX$ with $\neg\phi=G^{\nat}\circ F_{\nat}$.
\end{cor}

\begin{proof}
With Theorem \ref{Mphi_representation} at hand, it suffices to prove
$$\Kphi=\sM(\neg\phi).$$
Indeed, since $\{\neg\bbA\}_{\bbA\in\ob(\QDist)}$ is a dualizing family,
\begin{align*}
\phi_*\phi^*\lam&=((\neg\bbA\lda(\phi^*\lam))\rda\neg\bbA)\lda\phi\\
&=(\neg\bbA\lda(\lam\circ\phi))\rda(\neg\bbA\lda\phi)\\
&=((\neg\bbA\lda\phi)\lda\lam)\rda(\neg\bbA\lda\phi)\\
&=(\neg\phi\lda\lam)\rda\neg\phi\\
&=(\neg\phi)^{\da}(\neg\phi)_{\ua}\lam
\end{align*}
for all $\lam\in\PB$. Hence $\phi_*\phi^*=(\neg\phi)^{\da}(\neg\phi)_{\ua}:\PB\lra\PB$, and the conclusion follows.
\end{proof}

Although the above proof is indirect, it is not difficult prove this corollary directly using Theorems \ref{general_representation} and \ref{general_representation_dense_complete} as the following diagrams (cf. the diagrams \eqref{Mphi_representation_digram} below Theorem \ref{Mphi_representation}) sketch, which also explain the role of the codense $\CQ$-functor \eqref{neg_sYd_codense}:
\begin{equation} \label{Kphi_representation_Girard_diagram}
\bfig
\morphism(-800,0)<500,300>[\bbB`\PB;\sY_{\bbB}]
\morphism(800,0)<-500,300>[\bbA`\PA;\neg\sYd_{\bbA}]
\morphism(-800,0)|b|<800,-300>[\bbB`\Kphi;F=L\sY_{\bbB}]
\morphism(800,0)|b|<-800,-300>[\bbA`\Kphi;G=R\neg\sYd_{\bbA}]
\morphism(-300,300)|a|/@{>}@<4pt>/<600,0>[\PB`\PA;\phi^*]
\morphism(300,300)|b|/@{>}@<4pt>/<-600,0>[\PA`\PB;\phi_*]
\morphism(-300,300)|l|<300,-600>[\PB`\Kphi;L=\phi_*\phi^*]
\morphism(300,300)|r|<-300,-600>[\PA`\Kphi;R=\phi_*]
\place(0,305)[\mbox{\scriptsize$\bot$}]
\morphism(1100,0)<600,300>[\bbB`\PB;\sY_{\bbB}]
\morphism(2900,0)<-600,300>[\bbA`\PA;\neg\sYd_{\bbA}]
\morphism(1100,0)|b|<900,-300>[\bbB`\bbX;F]
\morphism(2900,0)|b|<-900,-300>[\bbA`\bbX;G]
\morphism(1700,300)|a|/@{>}@<4pt>/<600,0>[\PB`\PA;\phi^*]
\morphism(2300,300)|b|/@{>}@<4pt>/<-600,0>[\PA`\PB;\phi_*]
\morphism(1700,300)|l|<300,-600>[\PB`\bbX;L=\Lan_{\sY_{\bbB}}F]
\morphism(2300,300)|r|<-300,-600>[\PA`\bbX;R=\Ran_{\neg\sYd_{\bbA}}G]
\place(2000,305)[\mbox{\scriptsize$\bot$}]
\efig
\end{equation}

However, Corollary \ref{Kphi_representation_Girard} does not make sense for a general quantaloid $\CQ$. In fact, the following proposition blocks the way of finding a codense $\CQ$-functor $\bbA\lra\PA$ for an arbitrary $\CQ$-category $\bbA$ without assuming $\CQ$ being Girard:

\begin{prop} \label{codense_Girard}
Let $\CQ$ be a small quantaloid in which $1_q=\top_q:q\lra q$ for all $q\in\CQ_0$ and $\{\bot_q\}_{q\in\CQ_0}$ is a cyclic family, where $\top_q$ and $\bot_q$ denote the top and bottom arrows in $\CQ(q,q)$, respectively. Then the following statements are equivalent:
\begin{enumerate}[\rm (i)]
\item $\{\bot_q\}_{q\in\CQ_0}$ is a dualizing family, hence $\CQ$ is a Girard quantaloid.
\item There exists a codense $\CQ$-functor $F:\bbA\lra\PA$ for any $\CQ$-category $\bbA$.
\end{enumerate}
\end{prop}

\begin{proof}
(i)${}\Lra{}$(ii): $\neg\sYd_{\bbA}:\bbA\lra\PA$ is the required codense $\CQ$-functor (see \eqref{neg_sYd_codense}).

(ii)${}\Lra{}$(i): For each $q\in\CQ_0$, objects in $\sP\{q\}$ are exactly $\CQ$-arrows with domain $q$. Suppose there exists a codense $\CQ$-functor $F:\{q\}\lra\sP\{q\}$ targeting at $w:q\lra q$, then $F^{\nat}\in\sP\sP\{q\}$ satisfies $F^{\nat}(u)=w\lda u$ for all $u\in\sP\{q\}$, and consequently
\begin{align*}
\sP\{q\}(v,u)&=F^{\nat}(u)\rda F^{\nat}(v)&(\text{Proposition \ref{dense_condition}(iv)})\\
&=(w\lda u)\rda(w\lda v)\\
&=((w\lda u)\rda w)\lda v\\
&=\sP\{q\}(v,(w\lda u)\rda w)
\end{align*}
for all $v\in\sP\{q\}$, which implies $u=(w\lda u)\rda w$. In particular, since $1_q=\top_q$, letting $u=\bot_q$ and one has
$$w=1_q\rda w\leq(w\lda\bot_q)\rda w=\bot_q,$$
which exactly means $w=\bot_q$. Hence $u=(\bot_q\lda u)\rda\bot_q$, and the arbitrariness of $u$ indicates that $\{\bot_q\}_{q\in\CQ_0}$ is a dualizing family, completing the proof.
\end{proof}

A family of quantaloids that satisfy the hypotheses in Proposition \ref{codense_Girard} is given below:

\begin{exmp}
For any frame $(L,\wedge,\ra,0,1)$, one may construct a quantaloid $\DL$ \cite{Hohle2011,Pu2012,Walters1981} with the following data:
\begin{itemize}
\item objects in $\DL$ are the elements of $L$;
\item each $\DL(p,q)=\{u\in L:u\leq p\wedge q\}$ is equipped with the order inherited from $L$;
\item the composition of $\DL$-arrows $u\in\DL(p,q)$, $v\in\DL(q,r)$ is given by $v\circ u=v\wedge u$;
\item the implications of $\DL$-arrows are given by
$$w\lda u=q\wedge r\wedge(u\ra w)\quad\text{and}\quad v\rda w=p\wedge q\wedge(v\ra w)$$
for all $u\in\DL(p,q)$, $v\in\DL(q,r)$, $w\in\DL(p,r)$;
\item the identity $\DL$-arrow in $\DL(q,q)$ is $q$ itself.
\end{itemize}
It is straightforward to check that $\{0:q\lra q\}_{q\in L}$ is a cyclic family in $\DL$, but it is a dualizing family in $\DL$ if and only if $L$ is a Boolean algebra.
\end{exmp}

So, it is unavoidable that for a non-Girard quantaloid $\CQ$, a codense $\CQ$-subcategory of the presheaf $\CQ$-category $\PA$ would have a larger size than $\bbA$. If we look again at the codense $\CQ$-functor $\neg\sYd_{\bbA}:\bbA\lra\PA$ in \eqref{neg_sYd_codense} when $\CQ$ is Girard, we will see that it actually generates a codense $\CQ$-subcategory of $\PA$ with objects
\begin{equation} \label{oobbA_def}
\{\mu\in\PA\mid\mu=\neg\sYd_{\bbA}a=d_{|a|}\lda\bbA(a,-)\ \text{for some}\ a\in\bbA_0\};
\end{equation}
that is, presheaves on $\bbA$ which are \emph{complements of representable copresheaves} on $\bbA$. For a general quantaloid $\CQ$, although the complements of $\CQ$-distributors may not exist, \eqref{oobbA_def} suggests us to construct a $\CQ$-subcategory $\obbA$ of $\PA$ consisting of all the possible \emph{relative pseudo-complements of representable copresheaves} on $\bbA$:
\begin{equation} \label{obbA_def}
\obbA_0=\{\mu\in\PA\mid\mu=u\lda\bbA(a,-)\ \text{for some}\ a\in\bbA_0\ \text{and}\ \CQ\text{-arrow}\ u:|a|\lra\cod u\}.
\end{equation}
$\obbA$ is certainly a non-trivial $\CQ$-subcategory of $\PA$. We will see that $\obbA$ is a codense $\CQ$-subcategory of $\PA$ (Proposition \ref{obbA_codense}) and, moreover, a $\bw$-dense $\CQ$-subcategory of $\PA$ as an immediate consequence of Proposition \ref{UA_join_dense} discussed in Section \ref{Elementary}.

\subsection{Representation theorem for fixed points of Kan adjunctions}

Given a $\CQ$-category $\bbA$, the codomain restriction of the graph of the Yoneda embedding $(\sY_{\bbA})_{\nat}:\bbA\oto\PA$ on $\obbA$ gives a $\CQ$-distributor $\bbA^{\tr}:\bbA\oto\obbA$ with
\begin{equation} \label{Atr_mu}
\bbA^{\tr}(-,\mu)=(\sY_{\bbA})_{\nat}(-,\mu)=\mu
\end{equation}
for all $\mu\in\obbA_0$, where the second equality follows from the Yoneda lemma. The following identity holds for all $\CQ$-distributors $\phi:\bbA\oto\bbB$:

\begin{prop} \label{Atr_phi}
$\phi=(\bbA^{\tr}\lda\phi)\rda\bbA^{\tr}$.
\end{prop}

\begin{proof}
Since $\phi=(\bbA^{\tr}\lda\phi)\rda\bbA^{\tr}$ if and only if $\phi(-,b)=(\bbA^{\tr}\lda\phi(-,b))\rda\bbA^{\tr}$ for all $b\in\bbB_0$, it suffices to prove $\mu=(\bbA^{\tr}\lda\mu)\rda\bbA^{\tr}$ for any $\mu\in\PA$. On one hand,
$$\mu(a)\lda\bbA(a,-)\in\obbA_0$$
for all $a\in\bbA_0$ implies
\begin{align*}
\mu&=\mu\lda\bbA\\
&=\bw_{a\in\bbA_0}\mu(a)\lda\bbA(a,-)\\
&=\bw_{a\in\bbA_0}((\mu(a)\lda\mu(a))\rda\mu(a))\lda\bbA(a,-)\\
&=\bw_{a\in\bbA_0}(\mu(a)\lda\mu(a))\rda(\mu(a)\lda\bbA(a,-))\\
&=\bw_{a\in\bbA_0}(\mu(a)\lda(\mu\circ\bbA(a,-)))\rda(\mu(a)\lda\bbA(a,-))\\
&=\bw_{a\in\bbA_0}((\mu(a)\lda\bbA(a,-))\lda\mu)\rda(\mu(a)\lda\bbA(a,-))\\
&\geq\bw_{\mu'\in\obbA_0}(\mu'\lda\mu)\rda\mu'\\
&=\bw_{\mu'\in\obbA_0}(\bbA^{\tr}(-,\mu')\lda\mu)\rda\bbA^{\tr}(-,\mu')&(\text{Equation \eqref{Atr_mu}})\\
&=(\bbA^{\tr}\lda\mu)\rda\bbA^{\tr}.
\end{align*}
On the other hand, $\mu\leq(\bbA^{\tr}\lda\mu)\rda\bbA^{\tr}$ is trivial. The conclusion thus follows.
\end{proof}

This proposition indicates that the family $\{\bbA^{\tr}:\bbA\oto\obbA\}_{\bbA\in\ob(\QCat)}$ in $\QDist$ satisfies part of the properties of a dualizing family in a quantaloid (see \eqref{cyclic_dualizing_def}). Therefore, it makes sense to define the \emph{relative pseudo-complement} of any $\CQ$-distributor $\phi:\bbA\oto\bbB$ with respect to $\bbA^{\tr}$ as
\begin{equation} \label{phitr_def}
\phi^{\tr}:=\bbA^{\tr}\lda\phi:\bbB\oto\obbA.
\end{equation}
It is easy to obtain the following expressions for $\phi^{\tr}$:

\begin{lem} \label{phitr_calculation}
For any $\mu=u\lda\bbA(a,-)\in\obbA_0$,
\begin{enumerate}[\rm (1)]
\item $\phi^{\tr}(-,\mu)=\mu\lda\phi=\phi_*\mu$,
\item $\phi^{\tr}(-,u\lda\bbA(a,-))=u\lda\phi(a,-)$.
\end{enumerate}
\end{lem}

With the above preparations, now we are ready to establish the following representation theorem of $\Kphi$, which extends Corollary \ref{Kphi_representation_Girard} from a Girard quantaloid to a general quantaloid:

\begin{thm} \label{Kphi_representation}
For any $\CQ$-distributor $\phi:\bbA\oto\bbB$, a separated complete $\CQ$-category $\bbX$ is isomorphic to $\Kphi$ if, and only if, there exist a dense $\CQ$-functor $F:\bbB\lra\bbX$ and a codense $\CQ$-functor $G:\obbA\lra\bbX$ with $\phi^{\tr}=G^{\nat}\circ F_{\nat}$.
\end{thm}

\begin{proof}
Similar to the proof of Corollary \ref{Kphi_representation_Girard}, it suffices to prove $\phi_*\phi^*=(\phi^{\tr})^{\da}(\phi^{\tr})_{\ua}:\PB\lra\PB$, which implies $$\Kphi=\Mphi^{\tr}$$
and the conclusion would follow from Theorem \ref{Mphi_representation}. Indeed,
\begin{align*}
\phi_*\phi^*\lam&=((\bbA^{\tr}\lda(\phi^*\lam))\rda\bbA^{\tr})\lda\phi&(\text{Proposition \ref{Atr_phi}})\\
&=(\bbA^{\tr}\lda(\lam\circ\phi))\rda(\bbA^{\tr}\lda\phi)\\
&=((\bbA^{\tr}\lda\phi)\lda\lam)\rda(\bbA^{\tr}\lda\phi)\\
&=(\phi^{\tr}\lda\lam)\rda\phi^{\tr}&(\text{Equation \eqref{phitr_def}})\\
&=(\phi^{\tr})^{\da}(\phi^{\tr})_{\ua}\lam
\end{align*}
for all $\lam\in\PB$, as desired.
\end{proof}

Since $\sK\bbA=\PA$ (see Example \ref{MA_KA}), Theorem \ref{Kphi_representation} in particular implies the existence of a codense $\CQ$-functor $G:\obbA\lra\PA$. In fact, the inclusion $\CQ$-functor $J:\obbA\lra\PA$ is codense:

\begin{prop} \label{obbA_codense}
$\obbA$ is a codense $\CQ$-subcategory of $\PA$.
\end{prop}

\begin{proof}
For any $\mu\in\PA$, note that $\mu^{\tr}=\bbA^{\tr}\lda\mu:\{|\mu|\}\oto\obbA$ is in $\sPd\obbA$, and we claim $\mu=\lim_{\mu^{\tr}}J$. Indeed, by applying Equation \eqref{Atr_mu} to any $\lam\in\PA$ one has
\begin{equation} \label{Atr_lda_lam}
\bbA^{\tr}(-,\nu)\lda\lam=\nu\lda\lam=\PA(\lam,\nu)=\PA(\lam,J\nu)=J^{\nat}(\lam,\nu)
\end{equation}
for all $\nu\in\obbA_0$, and consequently
\begin{align*}
\PA(\lam,\mu)&=\mu\lda\lam\\
&=(\mu^{\tr}\rda\bbA^{\tr})\lda\lam&(\text{Proposition \ref{Atr_phi}})\\
&=\mu^{\tr}\rda(\bbA^{\tr}\lda\lam)\\
&=\mu^{\tr}\rda J^{\nat}(\lam,-),&(\text{Equation \eqref{Atr_lda_lam}})
\end{align*}
showing that $\mu=\lim_{\mu^{\tr}}J$.
\end{proof}

Hence, one is able to prove Theorem \ref{Kphi_representation} directly through Theorems \ref{general_representation} and \ref{general_representation_dense_complete} as the following diagrams illustrate; the details are similar to the proof of Theorem \ref{Mphi_representation} and will be left to the readers. We also remind the readers to compare \eqref{Kphi_representation_diagram} with the diagrams \eqref{Kphi_representation_Girard_diagram} under Corollary \ref{Kphi_representation_Girard}, which clearly indicate the difference between the cases of $\CQ$ being Girard or not:
\begin{equation} \label{Kphi_representation_diagram}
\bfig
\morphism(-800,0)<500,300>[\bbB`\PB;\sY_{\bbB}]
\morphism(800,0)/_(->/<-500,300>[\ \obbA`\PA;J]
\morphism(-800,0)|b|<800,-300>[\bbB`\Kphi;F=L\sY_{\bbB}]
\morphism(800,0)|b|<-800,-300>[\ \obbA`\Kphi;G=RJ]
\morphism(-300,300)|a|/@{>}@<4pt>/<600,0>[\PB`\PA;\phi^*]
\morphism(300,300)|b|/@{>}@<4pt>/<-600,0>[\PA`\PB;\phi_*]
\morphism(-300,300)|l|<300,-600>[\PB`\Kphi;L=\phi_*\phi^*]
\morphism(300,300)|r|<-300,-600>[\PA`\Kphi;R=\phi_*]
\place(0,305)[\mbox{\scriptsize$\bot$}]
\morphism(1100,0)<600,300>[\bbB`\PB;\sY_{\bbB}]
\morphism(2900,0)/_(->/<-600,300>[\ \obbA`\PA;J]
\morphism(1100,0)|b|<900,-300>[\bbB`\bbX;F]
\morphism(2900,0)|b|<-900,-300>[\ \obbA`\bbX;G]
\morphism(1700,300)|a|/@{>}@<4pt>/<600,0>[\PB`\PA;\phi^*]
\morphism(2300,300)|b|/@{>}@<4pt>/<-600,0>[\PA`\PB;\phi_*]
\morphism(1700,300)|l|<300,-600>[\PB`\bbX;L=\Lan_{\sY_{\bbB}}F]
\morphism(2300,300)|r|<-300,-600>[\PA`\bbX;R=\Ran_J G]
\place(2000,305)[\mbox{\scriptsize$\bot$}]
\efig
\end{equation}

\begin{rem}
The identity $\Kphi=\Mphi^{\tr}$ obtained in the proof of Theorem \ref{Kphi_representation} shows that the ``concept lattice'' of any multi-typed and multi-valued context in RST can be represented as the ``concept lattice'' of the relative pseudo-complement of the given context in FCA.  In fact, there are other trivial ways to represent any $\Kphi$ as a ``concept lattice'' in FCA.

First, for any separated complete $\CQ$-category $\bbA$ one may easily check
$$\sM(\bbA:\bbA\oto\bbA)=\sIm\sY_{\bbA}=\{\sY_{\bbA}a\mid a\in\bbA_0\}\cong\bbA.$$
In particular, $\Kphi$ is a separated complete $\CQ$-category and thus $\Kphi\cong\sM(\Kphi:\Kphi\oto\Kphi)$.

Second, similar to Proposition \ref{Atr_phi} one may prove that $\phi=((\sY_{\bbA})_{\nat}\lda\phi)\rda(\sY_{\bbA})_{\nat}$ for any $\CQ$-distributor $\phi:\bbA\oto\bbB$, and consequently one has
$$\Kphi=\sM((\sY_{\bbA})_{\nat}\lda\phi)$$
by performing the same calculations in Theorem \ref{Kphi_representation}, where $(\sY_{\bbA})_{\nat}\lda\phi$ is in fact the relative pseudo-complement of $\phi$ with respect to $(\sY_{\bbA})_{\nat}$.

Therefore, the point of the construction $\obbA$ in Theorem \ref{Kphi_representation} is to find a \emph{smallest possible} codense $\CQ$-subcategory of $\PA$. Although $\obbA$ may not be precisely the smallest one (e.g., when $\CQ$ is a Girard quantaloid), it is the best solution we find for an \emph{arbitrary} quantaloid $\CQ$.
\end{rem}

\begin{exmp}
For the identity $\CQ$-distributor $\bbA:\bbA\oto\bbA$ on a $\CQ$-category $\bbA$, since $\sK\bbA=\PA$, we have the dense Yoneda embedding $\sY_{\bbA}:\bbA\lra\PA$ and the codense inclusion $\CQ$-functor $J:\obbA\ \to/^(->/\PA$ that satisfy
$$\bbA^{\tr}(-,\mu)=\mu=\PA(\sY_{\bbA}-,\mu)=\PA(\sY_{\bbA}-,J\mu)=J^{\nat}(-,\mu)\circ(\sY_{\bbA})_{\nat}$$
for all $\mu\in\obbA_0$, where the first two equalities hold by Equation \eqref{Atr_mu} and the Yoneda lemma, respectively.

More generally, for any fully faithful $\CQ$-functor $F:\bbA\lra\bbB$ one has $\sK F^{\nat}=\PA$ since
$$(F^{\nat})_*(F^{\nat})^*\mu=(F_{\nat})^*(F^{\nat})^*\mu=(F^{\nat}\circ F_{\nat})^*\mu=\bbA^*\mu=\mu$$
for all $\mu\in\PA$, where the first and the third equalities respectively follow from Propositions \ref{adjoint_arrow_calculation}(1) and \ref{functor_graph}(1).
In this case, the dense Yoneda embedding $\sY_{\bbA}:\bbA\lra\PA$ and the codense $\CQ$-functor $H:=(\obbB\ \to/^(->/^J\PB\to^{(F^{\nat})_*}\PA)$ satisfy
$$(F^{\nat})^{\tr}(-,\lam)=(F^{\nat})_*\lam=\PA(\sY_{\bbA}-,H\lam)=H^{\nat}(-,\lam)\circ(\sY_{\bbA})_{\nat}$$
for all $\lam\in\obbB_0$, where the first two equalities follow from Lemma \ref{phitr_calculation}(1) and the Yoneda lemma, respectively.
\end{exmp}

\begin{exmp}
Let $\bbA$ be a \emph{completely distributive} (or equivalently, \emph{totally continuous}) $\CQ$-category \cite{Stubbe2007}; that is, a complete $\CQ$-category in which $\sup_{\bbA}:\PA\lra\bbA$ has a left adjoint $T_{\bbA}:\bbA\lra\PA$ in $\QCat$. Let $\theta_{\bbA}:\bbA\oto\bbA$ be the $\CQ$-distributor with transpose $\ola{\theta_{\bbA}}=T_{\bbA}$ (see \eqref{olphi_def}). Since $\sup_{\bbA}:\PA\lra\bbA$ is essentially surjective and thus codense (see Corollary \ref{composition_dense}(1)), from Proposition \ref{obbA_codense} and Corollary \ref{composition_dense}(2) one immediately knows that the restriction
$$G:=(\obbA\ \to/^(->/\PA\to^{\sup_{\bbA}}\bbA)$$
of $\sup_{\bbA}$ on $\obbA$ is codense. As $1_{\bbA}:\bbA\lra\bbA$ is obviously dense, and
\begin{align*}
\theta_{\bbA}^{\,\tr}(-,\mu)&=\mu\lda\theta_{\bbA}&(\text{Lemma \ref{phitr_calculation}(1)})\\
&=\PA(T_{\bbA}-,\mu)&(\ola{\theta_{\bbA}}=T_{\bbA})\\
&=\bbA(-,G\mu)&(T_{\bbA}\dv{\sup}_{\bbA})\\
&=G^{\nat}(-,\mu)\circ(1_{\bbA})_{\nat}
\end{align*}
for all $\mu\in\obbA_0$, one soon deduces $\sK\theta_{\bbA}\simeq\bbA$ from Theorem \ref{Kphi_representation}.
\end{exmp}

\subsection{The functoriality of relative pseudo-complements}

At the end of this section, we show that the identity
$$\Kphi=\sM\phi^{\tr}$$
obtained in the proof of Theorem \ref{Kphi_representation} can be established on the functorial level; that is, the process of generating the ``concept lattice'' in RST from a $\CQ$-distributor $\phi$ can be decomposed into two functorial steps:
\begin{enumerate}[\rm (1)]
\item calculating the relative pseudo-complement $\phi^{\tr}$;
\item generating the ``concept lattice'' of $\phi^{\tr}$ in FCA.
\end{enumerate}

First we establish the functoriality of the relative pseudo-complement
$$\phi^{\tr}=\bbA^{\tr}\lda\phi$$
of a $\CQ$-distributor $\phi$ with respect to $\bbA^{\tr}$. In fact, $\CQ$-distributors can be organized as objects into a category $\QChu$ with \emph{Chu transforms} (called \emph{infomorphisms} in \cite{Shen2013a})
$$(F,G):(\phi:\bbA\oto\bbB)\lra(\psi:\bbA'\oto\bbB')$$
as morphisms; that is, $\CQ$-functors $F:\bbA\lra\bbA'$ and $G:\bbB'\lra\bbB$ such that the square
$$\bfig
\square<600,500>[\bbA`\bbA'`\bbB`\bbB';F_\nat`\phi`\psi`G^\nat]
\place(300,0)[\circ] \place(300,500)[\circ] \place(0,250)[\circ] \place(600,250)[\circ]
\efig$$
is commutative, or equivalently, $\psi(F-,-)=\phi(-,G-)$.

For any $\CQ$-functor $F:\bbA\lra\bbA'$ and $u\lda\bbA(a,-)\in\obbA_0$, note that
\begin{equation} \label{F_nat_star_obbA}
(F_{\nat})_*(u\lda\bbA(a,-))=(u\lda\bbA(a,-))\lda F_{\nat}=u\lda(F_{\nat}\circ\bbA(a,-))=u\lda\bbA'(Fa,-)
\end{equation}
is an object in $\obbAp$. Thus $(F_{\nat})_*:\PA\lra\PA'$ can be restricted as a $\CQ$-functor $(F_{\nat})_*:\obbA\lra\obbAp$.

\begin{prop} \label{tl_functor}
$(G,(F_{\nat})_*):(\psi^{\tr}:\bbB'\oto\obbAp)\lra(\phi^{\tr}:\bbB\oto\obbA)$ is a Chu transform provided so is $(F,G):(\phi:\bbA\oto\bbB)\lra(\psi:\bbA'\oto\bbB')$.
\end{prop}

\begin{proof}
If $(F,G):\phi\lra\psi$ is a Chu transform, then for all $\mu=u\lda\bbA(a,-)\in\obbA_0$,
\begin{align*}
\phi^{\tr}(G-,\mu)&=u\lda\phi(a,G-)&\text{(Lemma \ref{phitr_calculation}(2))}\\
&=u\lda\psi(Fa,-)&((F,G)\ \text{is a Chu transform})\\
&=\psi^{\tr}(-,u\lda\bbA'(Fa,-))&\text{(Lemma \ref{phitr_calculation}(2))}\\
&=\psi^{\tr}(-,(F_{\nat})_*\mu).&\text{(Equation \eqref{F_nat_star_obbA})}
\end{align*}
Thus $(G,(F_{\nat})_*):\psi^{\tr}\lra\phi^{\tr}$ is a Chu transform.
\end{proof}

Proposition \ref{tl_functor} induces a contravariant functor
$$(-)^{\tr}:(\QChu)^{\op}\lra\QChu$$
that sends each $\CQ$-distributor to its relative pseudo-complement with respect to $\bbA^{\tr}$ and sends each Chu transform $(F,G):\phi\lra\psi$ to $(G,(F_{\nat})_*):\psi^{\tr}\lra\phi^{\tr}$.

It is known in \cite{Shen2013a} that the assignments $\phi\mapsto\Mphi$ and $\phi\mapsto\Kphi$ are respectively functorial and contravariant functorial from $\QChu$ to the category $\QSup$ of separated complete $\CQ$-categories and left adjoint $\CQ$-functors (or equivalently, $\sup$-preserving $\CQ$-functors; see Proposition \ref{la_condition}(iii)). Explicitly, for any Chu transform $(F,G):(\phi:\bbA\oto\bbB)\lra(\psi:\bbA'\oto\bbB')$,
$$\sM(F,G)=\dpsi\upsi(F^{\nat})^*:\Mphi\lra\sM\psi\quad\text{and}\quad\sK(F,G)=\phi_*\phi^*(G^{\nat})^*:\sK\psi\lra\Kphi$$
define functors
$$\sM:\QChu\lra\QSup\quad\text{and}\quad\sK:(\QChu)^{\op}\lra\QSup.$$
Hence, the identity $\Kphi=\Mphi^{\tr}$ can be expressed as the following commutative diagram:

\begin{prop} \label{Kphi=Mphitr_functor}
The diagram
$$\bfig
\qtriangle<1200,500>[(\QChu)^{\op}`\QChu`\QSup;(-)^{\tr}`\sK`\sM]
\efig$$
is commutative.
\end{prop}

\begin{proof}
It suffices to prove $\sM(G,(F_{\nat})_*)=\sK(F,G):\sK\psi=\sM\psi^{\tr}\lra\Kphi=\Mphi^{\tr}$ for any Chu transform $(F,G):\phi\lra\psi$. This is easy since
$$\sK(F,G)=\phi_*\phi^*(G^{\nat})^*=(\phi^{\tr})^{\da}(\phi^{\tr})_{\ua}(G^{\nat})^*=\sM(G,(F_{\nat})_*),$$
where the second equality holds because when restricting the codomain to the image, both $\phi_*\phi^*$ and $(\phi^{\tr})^{\da}(\phi^{\tr})_{\ua}$ are left adjoint to the same inclusion $\CQ$-functor $\Kphi=\Mphi^{\tr}\mkern6mu\to/^(->/\PB$ (see Proposition \ref{Q_closure_la}), thus they must be equal.
\end{proof}

\section{Elementary representation theorems in terms of join-(meet-)dense maps} \label{Elementary}

A map $f:A\lra B$ between (pre)ordered sets is \emph{$\bv$-dense} (i.e., \emph{join-dense}) (resp. \emph{$\bw$-dense} (i.e., \emph{meet-dense})) if, for any $y\in B$, there exists $\{x_i\}_{i\in I}\subseteq A$ with $y\cong\bv\limits_{i\in I}fx_i$ (resp. $y\cong\bw\limits_{i\in I}fx_i$). Obviously, $\bv$-dense (resp. $\bw$-dense) monotone maps between ordered sets are precisely dense (resp. codense) {\bf 2}-functors between {\bf 2}-categories. With a little abuse of language, we say that a $\CQ$-functor $F:\bbA\lra\bbB$ is $\bv$-dense (resp. $\bw$-dense) if its underlying type-preserving map $F:\bbA_0\lra\bbB_0$, as a monotone map between the underlying ordered sets of $\bbA$ and $\bbB$, is $\bv$-dense (resp. $\bw$-dense).

\begin{prop} \label{join_dense_imply_dense}
$\bv$-dense (resp. $\bw$-dense) $\CQ$-functors into complete $\CQ$-categories are necessarily dense (resp. codense).
\end{prop}

\begin{proof}
Let $F:\bbA\lra\bbB$ be a $\bv$-dense $\CQ$-functor, with $\bbB$ complete. For any $y\in\bbB_0$, let $\{x_i\}_{i\in I}\subseteq\bbA_0$ with $y\cong\bv\limits_{i\in I}Fx_i$, then $\bv\limits_{i\in I}\bbA(-,x_i)$, the join of $\{\bbA(-,x_i)\}_{i\in I}$ in the underlying order of $\PA$, is also in $\PA$. Since $\bbB$ is complete, $\sYd_{\bbB}:\PdB\lra\bbB$ is a left adjoint in $\QCat$ (see Theorem \ref{QCat_complete_equivalent}(vi)), and thus preserves underlying joins by Proposition \ref{la_condition}(4), i.e.,
$$\bbB\Big(\bv_{i\in I}Fx_i,-\Big)=\sYd_{\bbB}\bv_{i\in I}Fx_i=\bigsqcup_{i\in I}\sYd_{\bbB}Fx_i=\bw_{i\in I}\bbB(Fx_i,-),$$
where $\bigsqcup$ denotes the underlying join in $\PdB$ (see Remark \ref{PdA_QDist_order}). Hence
$$\bbB(y,-)=\bbB\Big(\bv_{i\in I}Fx_i,-\Big)=\bw_{i\in I}\bbB(Fx_i,-)=\bw_{i\in I}F_{\nat}\lda\bbA(-,x_i)=F_{\nat}\lda\bv_{i\in I}\bbA(-,x_i),$$
and consequently $y\cong\colim_{\bv\limits_{i\in I}\bbA(-,x_i)}F$, showing that $F$ is dense.
\end{proof}

\begin{rem}
The converse of Proposition \ref{join_dense_imply_dense} is not true; that is, dense (resp. codense) $\CQ$-functors are not necessarily $\bv$-dense (resp. $\bw$-dense). For example, the Yoneda embedding $\sY_{\bbA}:\bbA\lra\PA$ is dense for any $\CQ$-category $\bbA$ (see Example \ref{Yoneda_dense}), but it is not $\bv$-dense. In fact, this is clear when one considers the singleton $\CQ$-category $\{q\}$, in which case the image of $\sY_{\{q\}}:\{q\}\lra\sP\{q\}$ contains only one object and thus it can never be $\bv$-dense in $\sP\{q\}$ as long as $\CQ$ is larger than ${\bf 2}$. Similarly, the co-Yoneda embedding $\sYd_{\bbA}:\bbA\lra\PdA$ is codense but in general not $\bw$-dense.
\end{rem}

Each $\CQ$-typed set $A$ may be viewed as a \emph{discrete} $\CQ$-category with
$$A(x,y)=\begin{cases}
1_{|x|}, & \text{if}\ x=y,\\
\bot_{|x|,|y|}, & \text{else},
\end{cases}$$
where $\bot_{|x|,|y|}$ is the bottom arrow in $\CQ(|x|,|y|)$. It is easy to see that type-preserving maps from a discrete $\CQ$-category to any other $\CQ$-category are necessarily $\CQ$-functors. Therefore, Proposition \ref{join_dense_imply_dense} induces the following elementary version of Theorem \ref{general_representation_dense_complete} which only employs order-theoretic notions (i.e., $\bv$-density and $\bw$-density of maps) to characterize the $\CQ$-categorical equivalence:

\begin{thm} \label{general_representation_elementary}
Let $S\dv T:\bbD\lra\bbC$ be an adjunction between complete $\CQ$-categories. Then a complete $\CQ$-category $\bbX$ is equivalent to $\Fix(TS)$ if, and only if, there exist $\bv$-dense type-preserving maps $F:A\lra\bbX_0$, $K:A\lra\bbC_0$ and $\bw$-dense type-preserving maps $G:B\lra\bbX_0$, $H:B\lra\bbD_0$, where $A$, $B$ are $\CQ$-typed sets, such that $\bbD(SK-,H-)=\bbX(F-,G-)$.
\end{thm}

\begin{proof}
The necessity is trivial by taking $A=\bbC_0$, $B=\bbD_0$ and applying Theorem \ref{general_representation} as in the proof of Theorem \ref{general_representation_dense_complete}. For the sufficiency, the type-preserving maps $F,K,G,H$ are all $\CQ$-functors and, by Proposition \ref{distributor_graph}, $\bbX(F-,G-)=\bbD(SK-,H-)$ means precisely $G^{\nat}\circ F_{\nat}=H^{\nat}\circ S_{\nat}\circ K_{\nat}$:
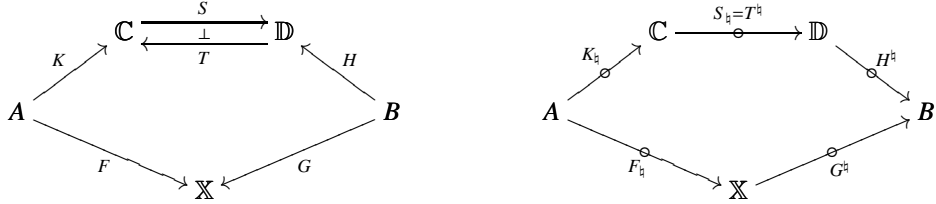
$$\bfig
\morphism(-700,0)<400,300>[A`\bbC;K]
\morphism(700,0)<-400,300>[B`\bbD;H]
\morphism(-700,0)|b|<700,-300>[A`\bbX;F]
\morphism(700,0)|b|<-700,-300>[B`\bbX;G]
\morphism(-300,300)|a|/@{>}@<4pt>/<600,0>[\bbC`\bbD;S]
\morphism(300,300)|b|/@{>}@<4pt>/<-600,0>[\bbD`\bbC;T]
\place(0,305)[\mbox{\scriptsize$\bot$}]
\morphism(1300,0)<400,300>[A`\bbC;K_{\nat}]
\morphism(1700,300)|a|<600,0>[\bbC`\bbD;S_{\nat}=T^{\nat}]
\morphism(2300,300)<400,-300>[\bbD`B;H^{\nat}]
\morphism(1300,0)|b|<700,-300>[A`\bbX;F_{\nat}]
\morphism(2000,-300)|b|<700,300>[\bbX`B;G^{\nat}]
\place(1500,150)[\circ] \place(2000,300)[\circ] \place(2500,150)[\circ] \place(1650,-150)[\circ] \place(2350,-150)[\circ]
\efig$$
Therefore, the conclusion follows soon from Proposition \ref{join_dense_imply_dense} and Theorem \ref{general_representation_dense_complete}.
\end{proof}

Theorem \ref{general_representation_elementary} is precisely the $\CQ$-version of Theorem \ref{general_representation_complete_lattice}. As applications of Theorem \ref{general_representation_elementary}, we will derive elementary representation theorems of $\Mphi$ and $\Kphi$ in the rest of this section. To this end, for any $\CQ$-category $\bbA$ we denote by $\bbA_0\times_{\dom}\CQ_1$ the $\CQ$-typed set
$$\bbA_0\times_{\dom}\CQ_1=\{(a,u)\mid a\in\bbA_0,\ u:|a|\lra\cod u\ \text{is a}\ \CQ\text{-arrow}\}$$
with types $|(a,u)|=\cod u$ for all $(a,u)\in\bbA_0\times_{\dom}\CQ_1$. Dually, we write
$$\bbA_0\times_{\cod}\CQ_1=\{(a,u)\mid a\in\bbA_0,\ u:\dom u\lra|a|\ \text{is a}\ \CQ\text{-arrow}\}$$
for the $\CQ$-typed set with types $|(a,u)|=\dom u$ for all $(a,u)\in\bbA_0\times_{\cod}\CQ_1$.

\begin{rem}
Note that neither $\bbA_0\times_{\dom}\CQ_1$ nor $\bbA_0\times_{\cod}\CQ_1$ is a product in the slice category $\Set\da\CQ_0$. Let $(\CQ_1,\dom)$ and $(\CQ_1,\cod)$ be $\CQ$-typed sets with type maps sending each $\CQ$-arrow to its domain and codomain, respectively. Then the product $\bbA_0\times(\CQ_1,\dom)$ in $\Set\da\CQ_0$ has exactly the same underlying set as $\bbA_0\times_{\dom}\CQ_1$, but the type of $(a,u)\in\bbA_0\times(\CQ_1,\dom)$ is $|(a,u)|=|a|=\dom u$. Similarly, the underlying set of the product $\bbA_0\times(\CQ_1,\cod)$ in $\Set\da\CQ_0$ is the same as $\bbA_0\times_{\cod}\CQ_1$, but the type of $(a,u)\in\bbA_0\times(\CQ_1,\cod)$ is $|(a,u)|=|a|=\cod u$.
\end{rem}

Considering $\bbA_0\times_{\dom}\CQ_1$ and $\bbA_0\times_{\cod}\CQ_1$ as discrete $\CQ$-categories, one has the following $\CQ$-functors (which only have to be type-preserving maps):
\renewcommand\arraystretch{1.5}
$$\begin{array}{ll}
U_{\bbA}:\bbA_0\times_{\dom}\CQ_1\lra\PA, & U_{\bbA}(a,u)=u\circ\sY_{\bbA}a=u\circ\bbA(-,a),\\
N_{\bbA}:\bbA_0\times_{\dom}\CQ_1\lra\PA, & N_{\bbA}(a,u)=u\lda\sYd_{\bbA}a=u\lda\bbA(a,-),\\
\Ud_{\bbA}:\bbA_0\times_{\cod}\CQ_1\lra\PdA, & \Ud_{\bbA}(a,u)=\sYd_{\bbA}a\circ u=\bbA(a,-)\circ u,\\
\Nd_{\bbA}:\bbA_0\times_{\cod}\CQ_1\lra\PdA, & \Nd_{\bbA}(a,u)=\sY_{\bbA}a\rda u=\bbA(-,a)\rda u.
\end{array}$$

\begin{prop} \label{UA_join_dense}
For any $\CQ$-category $\bbA$, $U_{\bbA}$, $\Nd_{\bbA}$ are $\bv$-dense, and $N_{\bbA}$, $\Ud_{\bbA}$ are $\bw$-dense.
\end{prop}

\begin{proof}
The $\bv$-density of $U_{\bbA}$ and the $\bw$-density of $N_{\bbA}$ are easy since
\begin{align*}
&\mu=\mu\circ\bbA=\bv_{a\in\bbA_0}\mu(a)\circ\bbA(-,a)=\bv_{a\in\bbA_0}U_{\bbA}(a,\mu(a)),\\
&\mu=\mu\lda\bbA=\bw_{a\in\bbA_0}\mu(a)\lda\bbA(a,-)=\bw_{a\in\bbA_0}N_{\bbA}(a,\mu(a))
\end{align*}
for all $\mu\in\PA$. For the $\bw$-density of $\Ud_{\bbA}$ and the $\bv$-density of $\Nd_{\bbA}$, just note that
\begin{align*}
&\lam=\bbA\circ\lam=\bv_{a\in\bbA_0}\bbA(a,-)\circ\lam(a)=\bigsqcap_{a\in\bbA_0}\Ud_{\bbA}(a,\lam(a)),\\
&\lam=\bbA\rda\lam=\bw_{a\in\bbA_0}\bbA(-,a)\rda\lam(a)=\bigsqcup_{a\in\bbA_0}\Nd_{\bbA}(a,\lam(a))
\end{align*}
for all $\lam\in\PdA$, where $\bigsqcap$ and $\bigsqcup$ are calculated in the underlying order of $\PdA$ (see Remark \ref{PdA_QDist_order}).
\end{proof}

It is easy to observe $\sIm(N_{\bbA})=\obbA$ (see \eqref{obbA_def}), and thus, as we mentioned at the end of Subsection \ref{Codense_Sub_Presheaf}, the crucial construction $\obbA$ in the representation theorem of $\Kphi$ (i.e., Theorem \ref{Kphi_representation}) is in fact a $\bw$-dense $\CQ$-subcategory of $\PA$.

\begin{prop} \label{v_rda_phiab_lda_u}
Let $\phi:\bbA\oto\bbB$ be a $\CQ$-distributor. Then
\begin{enumerate}[\rm (1)]
\item $(\uphi)_{\nat}(U_{\bbA}(a,u),\Ud_{\bbB}(b,v))=v\rda(\phi(a,b)\lda u)$ for all $(a,u)\in\bbA_0\times_{\dom}\CQ_1$, $(b,v)\in\bbB_0\times_{\cod}\CQ_1$,
\item $(\phi^*)_{\nat}(U_{\bbB}(b,v),N_{\bbA}(a,u))=(u\lda\phi(a,b))\lda v$ for all $(a,u)\in\bbA_0\times_{\dom}\CQ_1$, $(b,v)\in\bbB_0\times_{\dom}\CQ_1$.
\end{enumerate}
\end{prop}

\begin{proof}
Straightforward calculation.
%
\end{proof}

\begin{thm} \label{Mphi_representation_elementary}
For any $\CQ$-distributor $\phi:\bbA\oto\bbB$, a separated complete $\CQ$-category $\bbX$ is isomorphic to $\Mphi$ if, and only if, there exist a $\bv$-dense type-preserving map $F:\bbA_0\times_{\dom}\CQ_1\lra\bbX_0$ and a $\bw$-dense type-preserving map $G:\bbB_0\times_{\cod}\CQ_1\lra\bbX_0$ such that
$$v\rda(\phi(a,b)\lda u)=\bbX(F(a,u),G(b,v))$$
for all $(a,u)\in\bbA_0\times_{\dom}\CQ_1$, $(b,v)\in\bbB_0\times_{\cod}\CQ_1$.
\end{thm}

\begin{proof}
{\bf Necessity.} By Theorem \ref{general_representation} there exist surjective $\CQ$-functors $L:\PA\lra\bbX$ and $R:\PdB\lra\bbX$ with $(\uphi)_{\nat}=\bbX(L-,R-)$. It is easy to see that Corollary \ref{composition_dense} also holds for $\bv$-dense and $\bw$-dense type-preserving maps; in fact, one just needs to consider $\CQ={\bf 2}$ and note that left (resp. right) adjoint $\CQ$-functors are also left (resp. right) adjoints in the underlying order (see Proposition \ref{la_condition}(4)). Therefore, following the same reasoning in the proof of Theorem \ref{Mphi_representation} one deduces the $\bv$-density of $LU_{\bbA}:\bbA_0\times_{\dom}\CQ_1\lra\bbX_0$ and the $\bw$-density of $R\Ud_{\bbB}:\bbB_0\times_{\cod}\CQ_1\lra\bbX_0$.
$$\bfig
\morphism(-1200,0)<900,300>[\bbA_0\times_{\dom}\CQ_1`\PA;U_{\bbA}]
\morphism(1200,0)<-900,300>[\bbB_0\times_{\cod}\CQ_1`\PdB;\Ud_{\bbB}]
\morphism(-1200,0)|b|<1200,-300>[\bbA_0\times_{\dom}\CQ_1`\Mphi;F=LU_{\bbA}]
\morphism(1200,0)|b|<-1200,-300>[\bbB_0\times_{\cod}\CQ_1`\Mphi;G=R\Ud_{\bbB}]
\morphism(-300,300)|a|/@{>}@<4pt>/<600,0>[\PA`\PdB;\uphi]
\morphism(300,300)|b|/@{>}@<4pt>/<-600,0>[\PdB`\PA;\dphi]
\morphism(-300,300)|l|<300,-600>[\PA`\Mphi;L=\dphi\uphi]
\morphism(300,300)|r|<-300,-600>[\PdB`\Mphi;R=\dphi]
\place(0,305)[\mbox{\scriptsize$\bot$}]
\efig$$
Finally, Proposition \ref{v_rda_phiab_lda_u}(1) implies
$$v\rda(\phi(a,b)\lda u)=(\uphi)_{\nat}(U_{\bbA}(a,u),\Ud_{\bbB}(b,v))=\bbX(LU_{\bbA}(a,u),R\Ud_{\bbB}(b,v))$$
for all $(a,u)\in\bbA_0\times_{\dom}\CQ_1$, $(b,v)\in\bbB_0\times_{\cod}\CQ_1$.

{\bf Sufficiency.} From Proposition \ref{UA_join_dense} we have $\bv$-dense type-preserving maps $U_{\bbA}:\bbA_0\times_{\dom}\CQ_1\lra\PA$, $F:\bbA_0\times_{\dom}\CQ_1\lra\bbX_0$ and $\bw$-dense type-preserving maps $\Ud_{\bbB}:\bbB_0\times_{\cod}\CQ_1\lra\PdB$, $G:\bbB_0\times_{\cod}\CQ_1\lra\bbX_0$ with
$$\bbX(F(a,u),G(b,v))=v\rda(\phi(a,b)\lda u)=\PdB(\uphi U_{\bbA}(a,u),\Ud_{\bbB}(b,v))
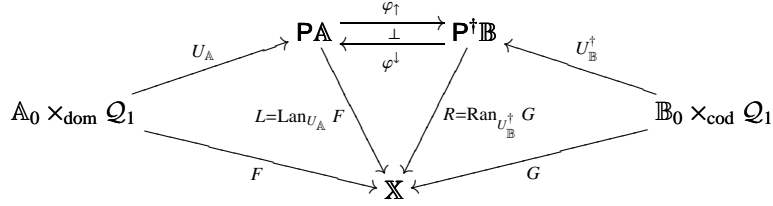
$$
$$\bfig
\morphism(-1200,0)<900,300>[\bbA_0\times_{\dom}\CQ_1`\PA;U_{\bbA}]
\morphism(1200,0)<-900,300>[\bbB_0\times_{\cod}\CQ_1`\PdB;\Ud_{\bbB}]
\morphism(-1200,0)|b|<1200,-300>[\bbA_0\times_{\dom}\CQ_1`\bbX;F]
\morphism(1200,0)|b|<-1200,-300>[\bbB_0\times_{\cod}\CQ_1`\bbX;G]
\morphism(-300,300)|a|/@{>}@<4pt>/<600,0>[\PA`\PdB;\uphi]
\morphism(300,300)|b|/@{>}@<4pt>/<-600,0>[\PdB`\PA;\dphi]
\morphism(-300,300)|l|<300,-600>[\PA`\bbX;L=\Lan_{U_{\bbA}}F]
\morphism(300,300)|r|<-300,-600>[\PdB`\bbX;R=\Ran_{\Ud_{\bbB}}G]
\place(0,305)[\mbox{\scriptsize$\bot$}]
\efig$$
for all $(a,u)\in\bbA_0\times_{\dom}\CQ_1$, $(b,v)\in\bbB_0\times_{\cod}\CQ_1$ by Proposition \ref{v_rda_phiab_lda_u}(1). Thus the conditions in Theorem \ref{general_representation_elementary} are satisfied, completing the proof.
\end{proof}

\begin{thm} \label{Kphi_representation_elementary}
For any $\CQ$-distributor $\phi:\bbA\oto\bbB$, a separated complete $\CQ$-category $\bbX$ is isomorphic to $\Kphi$ if, and only if, there exist a $\bv$-dense type-preserving map $F:\bbB_0\times_{\dom}\CQ_1\lra\bbX_0$ and a $\bw$-dense type-preserving map $G:\bbA_0\times_{\dom}\CQ_1\lra\bbX_0$ such that
$$(u\lda\phi(a,b))\lda v=\bbX(F(b,v),G(a,u))$$
for all $(a,u)\in\bbA_0\times_{\dom}\CQ_1$, $(b,v)\in\bbB_0\times_{\dom}\CQ_1$.
\end{thm}

\begin{proof}
Similar to Theorem \ref{Mphi_representation_elementary} under the help of Proposition \ref{v_rda_phiab_lda_u}(2) and the details are left to the readers. Here we just sketch the diagrams both for the ``only if'' part and the ``if'' part as a comparison to the above theorem and the diagrams \eqref{Kphi_representation_diagram} illustrating Theorem \ref{Kphi_representation}:
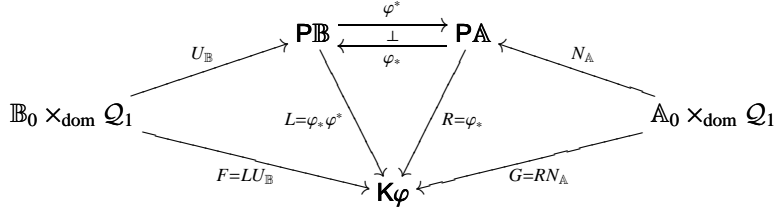
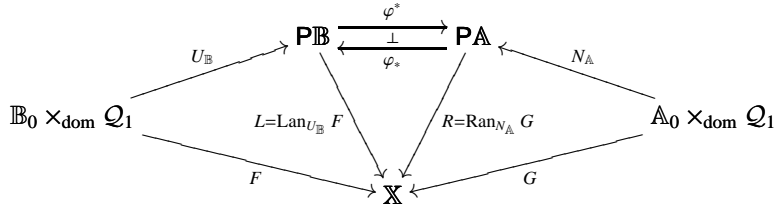
$$\bfig
\morphism(-1200,0)<900,300>[\bbB_0\times_{\dom}\CQ_1`\PB;U_{\bbB}]
\morphism(1200,0)<-900,300>[\bbA_0\times_{\dom}\CQ_1`\PA;N_{\bbA}]
\morphism(-1200,0)|b|<1200,-300>[\bbB_0\times_{\dom}\CQ_1`\Kphi;F=LU_{\bbB}]
\morphism(1200,0)|b|<-1200,-300>[\bbA_0\times_{\dom}\CQ_1`\Kphi;G=RN_{\bbA}]
\morphism(-300,300)|a|/@{>}@<4pt>/<600,0>[\PB`\PA;\phi^*]
\morphism(300,300)|b|/@{>}@<4pt>/<-600,0>[\PA`\PB;\phi_*]
\morphism(-300,300)|l|<300,-600>[\PB`\Kphi;L=\phi_*\phi^*]
\morphism(300,300)|r|<-300,-600>[\PA`\Kphi;R=\phi_*]
\place(0,305)[\mbox{\scriptsize$\bot$}]
\efig$$
$$\bfig
\morphism(-1200,0)<900,300>[\bbB_0\times_{\dom}\CQ_1`\PB;U_{\bbB}]
\morphism(1200,0)<-900,300>[\bbA_0\times_{\dom}\CQ_1`\PA;N_{\bbA}]
\morphism(-1200,0)|b|<1200,-300>[\bbB_0\times_{\dom}\CQ_1`\bbX;F]
\morphism(1200,0)|b|<-1200,-300>[\bbA_0\times_{\dom}\CQ_1`\bbX;G]
\morphism(-300,300)|a|/@{>}@<4pt>/<600,0>[\PB`\PA;\phi^*]
\morphism(300,300)|b|/@{>}@<4pt>/<-600,0>[\PA`\PB;\phi_*]
\morphism(-300,300)|l|<300,-600>[\PB`\bbX;L=\Lan_{U_{\bbB}}F]
\morphism(300,300)|r|<-300,-600>[\PA`\bbX;R=\Ran_{N_{\bbA}} G]
\place(0,305)[\mbox{\scriptsize$\bot$}]
\efig$$
\end{proof}

In the case that $\CQ$ has only one object, i.e., a \emph{unital quantale}, both $\bbA_0\times_{\dom}\CQ_1$ and $\bbA_0\times_{\cod}\CQ_1$ become the cartesian product of the set $\bbA_0$ and the set of elements of $\CQ$. As the following immediate corollary of Theorems \ref{Mphi_representation_elementary} and \ref{Kphi_representation_elementary} states, our results generalize B{\v e}lohl{\' a}vek's representation theorem for concept lattices of quantale-valued contexts in FCA (see \cite[Theorem 14(2)]{Bvelohlavek2004}) and Popescu's representation theorem for those in RST (see \cite[Proposition 7.3]{Popescu2004}):

\begin{cor} \label{Mphi_Kphi_representation_quantale}
Let $\CQ$ be a unital quantale, $\phi:\bbA\oto\bbB$ a $\CQ$-distributor and $\bbX$ a separated complete $\CQ$-category.
\begin{enumerate}[\rm (1)]
\item $\bbX$ is isomorphic to $\Mphi$ if, and only if, there exist a $\bv$-dense map $F:\bbA_0\times\CQ\lra\bbX_0$ and a $\bw$-dense map $G:\bbB_0\times\CQ\lra\bbX_0$ such that
    $$v\rda(\phi(a,b)\lda u)=\bbX(F(a,u),G(b,v))$$
    for all $a\in\bbA_0$, $b\in\bbB_0$, $u,v\in\CQ$.
\item $\bbX$ is isomorphic to $\Kphi$ if, and only if, there exist a $\bv$-dense map $F:\bbB_0\times\CQ\lra\bbX_0$ and a $\bw$-dense map $G:\bbA_0\times\CQ\lra\bbX_0$ such that
    $$(u\lda\phi(a,b))\lda v=\bbX(F(b,v),G(a,u))$$
    for all $a\in\bbA_0$, $b\in\bbB_0$, $u,v\in\CQ$.
\end{enumerate}
\end{cor}

\begin{rem} \label{Belohlavek_Popescu}
B{\v e}lohl{\' a}vek's \cite[Theorem 14(2)]{Bvelohlavek2004} is precisely Corollary \ref{Mphi_Kphi_representation_quantale}(1) when $\CQ$ is a commutative integral quantale, while Popescu's \cite[Proposition 7.3]{Popescu2004} is a weaker version of our Corollary \ref{Mphi_Kphi_representation_quantale}(2) even if $\CQ$ is commutative and integral. Explicitly, Popescu's result should be stated as:
\begin{quote}
A complete lattice $X$ is isomorphic to the underlying complete lattice of $\Kphi$ if, and only if, there exist a $\bv$-dense map $F:\bbB_0\times\CQ\lra X$ and a $\bw$-dense map $G:\bbA_0\times\CQ\lra X$ such that
$$\phi(a,b)\leq v\rda u\iff F(b,v)\leq G(a,u)$$
for all $a\in\bbA_0$, $b\in\bbB_0$, $u,v\in\CQ$.
\end{quote}
In fact, the ``only if'' part of the above claim is an immediate consequence of Corollary \ref{Mphi_Kphi_representation_quantale}(2), and the ``if'' part follows by applying Theorem \ref{general_representation_elementary} in the case $\CQ={\bf 2}$ to the underlying Galois connection (between the underlying ordered sets) of $\phi^*\dv\phi_*$.
\end{rem}

\section{Concluding remarks}

The following diagram indicates the connections between the most important representation theorems established in this paper. With the general representation theorems for fixed points of adjoint $\CQ$-functors in hand, one is able to derive various kinds of representation theorems for concept lattices in FCA and those in RST in the generality their $\CQ$-version. Besides unifying and extending existing representation theorems, their most important application in this paper is the development of a universal approach to represent concept lattices in RST as those in FCA; to our knowledge, this has never been achieved before for multi-typed and multi-valued contexts possibly valued in a non-Girard quantaloid, although its {\bf 2}-version is trivial. Moreover, we believe that the general representation theorems have the potential to be applied to more areas which deserve further investigation.

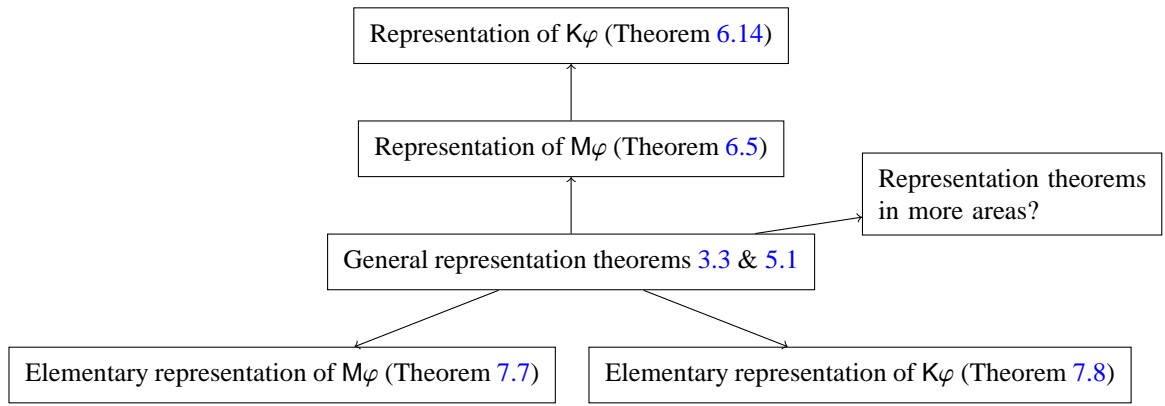
\tikzset{
box/.style={rectangle,
minimum width=30pt, minimum height=15pt, inner sep=6pt,draw}
}
\begin{center}
\begin{tikzpicture}
\node[box] (g) at (0,0) {General representation theorems \ref{general_representation} \& \ref{general_representation_dense_complete}};
\node[box] (M) at (0,1.5) {Representation of $\Mphi$ (Theorem \ref{Mphi_representation})};
\node[box] (K) at (0,3) {Representation of $\Kphi$ (Theorem \ref{Kphi_representation})};
\node[box] (Me) at (-3.8,-1.5) {Elementary representation of $\Mphi$ (Theorem \ref{Mphi_representation_elementary})};
\node[box] (Ke) at (3.8,-1.5) {Elementary representation of $\Kphi$ (Theorem \ref{Kphi_representation_elementary})};
\node[box,text width=10em] (f2) at (5.8,0.9) {Representation theorems in more areas?};
\draw[->] (g)--(M);
\draw[->] (M)--(K);
\draw[->] (g)--(Me);
\draw[->] (g)--(Ke);
\draw[->] (g)--(f2);
\end{tikzpicture}
\end{center}

\section*{Acknowledgement}

The first author acknowledges the support of National Natural Science Foundation of China (11101297) and International Visiting Program for Excellent Young Scholars of Sichuan University. He is also grateful to Professor Walter Tholen for warm hospitality during his visit to the Department of Mathematics and Statistics at York University where he completes this paper.

The second author acknowledges the support of Natural Sciences and Engineering Research Council of Canada (Discovery Grant 501260 held by Professor Walter Tholen).





\end{document}